\documentclass[12pt]{article}
\usepackage{amsmath}
\usepackage{graphicx}
\usepackage{natbib}
\usepackage{url} 

\usepackage{array}
\usepackage{graphicx}
\usepackage{multirow}
\usepackage{amsmath,amssymb,amsfonts}
\usepackage{amsthm}
\usepackage{mathrsfs}
\usepackage{subfigure}
\usepackage{manyfoot}

\usepackage{mathtools}

\DeclarePairedDelimiter\norm{\lVert}{\rVert}

\usepackage{bm}

\usepackage{tikz}
\usepackage{pgfplots}
\pgfplotsset{compat=newest}
\usetikzlibrary{shapes.geometric,arrows,fit,matrix,positioning}
\tikzset
{
    treenode/.style = {rectangle, draw=black, align=center, minimum size=1cm},
    terminal/.style  = {circle, draw=black, align=center, minimum height=0.8cm, minimum width=1cm, anchor=north}
}

\newcommand{\R}{\mathbb{R}}
\newcommand{\E}{\mathbb{E}}
\newcommand{\V}{\mathbb{V}}
\newcommand{\x}{\mathbf{x}}

\newcommand{\X}{\mathbf{X}}

\newcommand{\cA}{\mathcal{A}}
\newcommand{\cC}{\mathcal{C}}

\newcommand{\cZ}{\mathcal{Z}}
\newcommand{\cH}{\mathcal{H}}

\newcommand{\sX}{\mathscr{X}}

\newcommand{\mfY}{\mathfrak{Y}}
\newcommand{\bbR}{\mathbb{R}}

\newcommand{\perm}{\tau}

\newcommand{\T}{\mathcal{T}}

\newcommand{\balpha}{\bar{\alpha}}
\newcommand{\bmalpha}{\bm{\alpha}}

\newtheorem{theorem}{Theorem}
\newtheorem{lemma}{Lemma}
\newtheorem{corollary}{Corollary}

\newtheorem{definition}{Definition}

\date{}

\begin{document}

\def\spacingset#1{\renewcommand{\baselinestretch}
{#1}\small\normalsize} \spacingset{1}


\title{\bf Estimating Shapley Effects in Big-Data Emulation and Regression Settings using Bayesian Additive Regression Trees}
  \author{Akira Horiguchi\thanks{
      AH would like to acknowledge Miheer Dewaskar for fruitful discussions.  The work of MTP was supported in part by the National Science Foundation under Agreements DMS-1916231, OAC-2004601, and in part by the King Abdullah University of Science and Technology (KAUST) Office of Sponsored Research (OSR) under Award No. OSR-2018-CRG7-3800.3.
    }\hspace{.2cm}\\
    Department of Statistics, University of California, Davis\\
    and \\
    Matthew T. Pratola \\
    Department of Statistics, Indiania University Bloomington}
  \maketitle

\bigskip
\begin{abstract}
Shapley effects are a particularly interpretable approach to assessing how a function depends on its various inputs. 
The existing literature contains various estimators for this class of sensitivity indices in the context of nonparametric regression where the function is observed with noise, but there does not seem to be an estimator that is computationally tractable for input dimensions in the hundreds scale.
This article provides such an estimator that is computationally tractable on this scale.
The estimator uses a metamodel-based approach by first fitting a Bayesian Additive Regression Trees model which is then used to compute Shapley-effect estimates.
This article also establishes a theoretical guarantee of posterior consistency on a large function class for this Shapley-effect estimator.
Finally, this paper explores the performance of these Shapley-effect estimators on four different test functions for various input dimensions, including $p=500$.
\end{abstract}

\vfill

\newpage

\section{Introduction}
\label{sec:introduction}

An important task in global sensitivity analysis is to measure how a real-valued function depends on its various inputs. 
A popular measure of variable importance is the class of Sobol\'{} indices \citep{sobol1990sensitivity}, which decomposes the variance of outputs from a function into terms due to main effects for each input and interaction effects between the various inputs. 
To quantify the impact of any particular input dimension, either the \textit{main-effect Sobol\'{} index} or the \textit{total-effect Sobol\'{} index} can be used; 
the latter includes all interactions between the given input and any other input whereas the former excludes any such interaction.
Straightforward interpretation of Sobol\'{} indices requires an orthogonal distribution on the inputs \citep{song2016shapley}.
\textit{Shapley effects} \citep{shapley1952,song2016shapley} form another class of variance-based global sensitivity indices that was first introduced in the context of game theory but has only recently been gaining traction in the statistics literature \citep{owen2014sobol}.
Although the additional computation required to compute Shapley effects might render them unnecessary if the inputs are known to be independent, Shapley effects remain interpretable even if the inputs are correlated \citep{song2016shapley} and hence are the more reasonable option in such a case. 

If the function of interest is known and has a simple enough form, its exact Shapley effects can sometimes be computed analytically, particularly when the required integrals can be computed easily.
Otherwise, the Shapley effects can be estimated using values generated from the function. Many existing methods assume the function can be evaluated cheaply and without observation noise and indeed work well in such a scenario.
Figure~\ref{fig:ppMC} shows various such Shapley-effect estimators \citep{song2016shapley,benoumechiara2019shapley,broto2020variance,plischke2021computing,goda2021simple} applied to $n$ observations generated from a function (defined in the figure caption) evaluated on i.i.d.\ inputs drawn uniformly from the hypercube $[0,1]^{5}$. 
When the function values are observed without noise, these methods track the $g$-function's true Shapley-effects very well. 
But when independent and identically distributed (i.i.d.) Gaussian noise with mean zero and moderate variance (defined in the figure caption) is added, these methods struggle to capture the true values even when the number of observations increases dramatically to compensate for the observation noise.  

\begin{figure}[h!]
\centering
\includegraphics[width=0.9\textwidth]{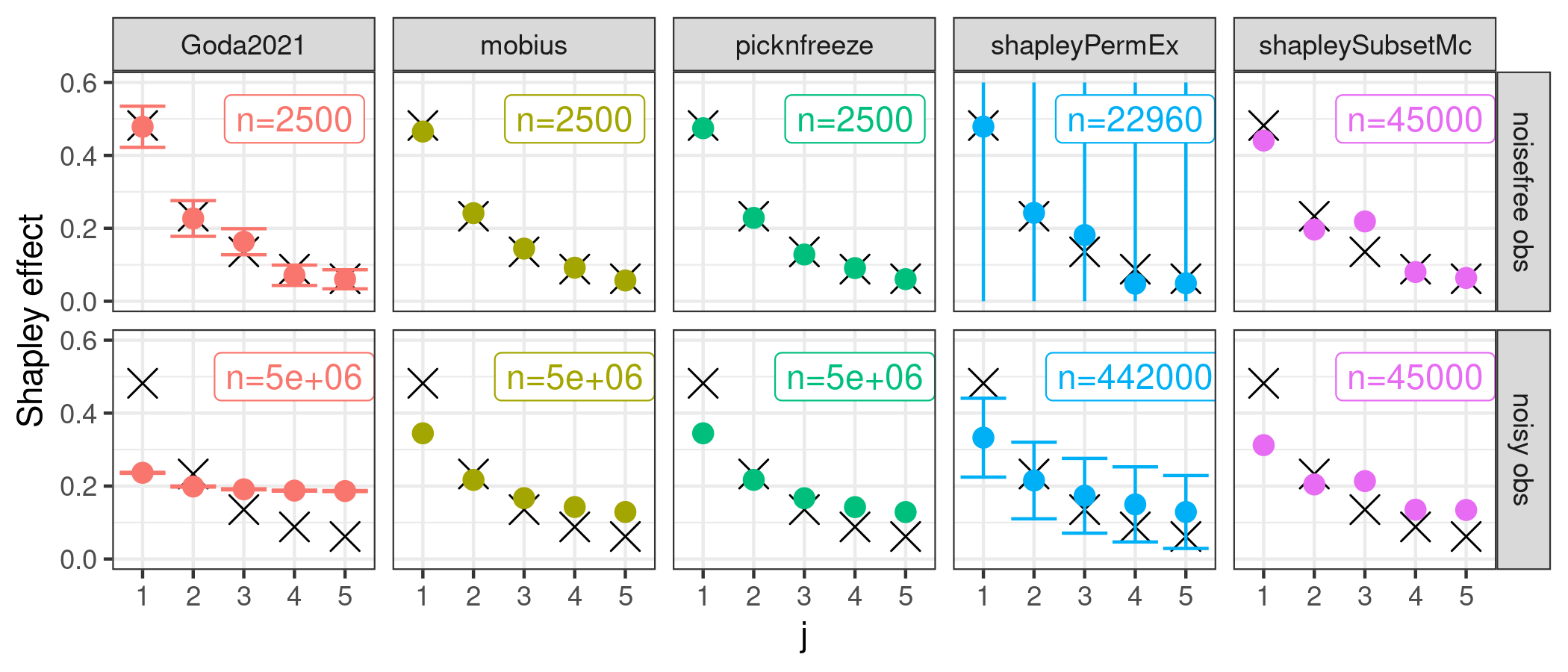}
\caption{Shapley-effect estimates of various existing methods trained on data drawn from the Sobol\'{} $g$-function $f(\x)=\prod_{k=1}^{5} \frac{|4x_k-2| + (k-1)/2}{1+(k-1)/2}$. Crosses represent the true Shapley-effect values. 
Function values are evaluated 
at $n$ i.i.d.\ inputs drawn uniformly from the hypercube $[0,1]^5$. In the top row, function values are observed without noise. In the bottom row, the observations are function values plus i.i.d.\ Gaussian noise with mean zero and variance $0.25\times3.076$, where $3.076$ is the variance of the $g$-function under a uniform distribution on $[0,1]^5$. Each column represents an estimation method: ``Goda2021'' is from \cite{goda2021simple}; ``mobius'' and ``picknfreeze'' are from \cite{plischke2021computing}; ``shapleyPermEx'' and ``shapleySubsetMc'' are from \cite{sensitivityRpackage}. Error bars represent approximate or exact $95\%$ confidence intervals as implemented by the method which aim to capture the variability induced by the Monte Carlo approximation of expectations.}
\label{fig:ppMC}
\end{figure}

For noisy function observations, one can first estimate the function and then compute sensitivity indices of the estimated function as a post-processing step.
One option is to fit a metamodel to the observations; the fitted metamodel then serves as the estimated function. (This approach is also useful in noisefree settings when the function can only be sparsely evaluated and the fitted metamodel can be evaluated cheaply.)
Popular metamodels include 
the Gaussian Process (GP), 
Bayesian multivariate adaptive regression splines (BMARS) \citep{denison1998bayesian}, 
generalized polynomial chaos expansions (PCE) \citep{Sudret2008}, 
treed GPs \cite{Gramacy2010}, 
dynamic trees \citep{Gramacy2013}, 
Gaussian radial basis function \citep{wu2016global}, 
artificial neural networks \citep{li2016accelerate},
and deep GPs \citep{radaideh2020surrogate}.
This paper makes its contributions using Bayesian Additive Regression Trees (BART) \citep{chipman2010bart} which is an increasingly popular tool for complex regression problems and as emulators of expensive computer simulations \citep{Chipman12,Gramacy16,horiguchi2022using}.
BART is a nonparametric sum-of-trees model embedded in a Bayesian inferential framework. 
Unlike many other metamodels, BART can easily incorporate categorical inputs, avoids strong parametric assumptions, and is relatively quick to fit even on a large number of observations. 
BART even has been shown to be resilient to the inclusion of inert inputs, particularly when the BART prior incorporates either the sparsity-inducing Dirichlet prior of \cite{Linero18} or the spike-and-tree prior of \cite{vanderPasRockova17,LRW18}.
Furthermore, the Bayesian framework provides natural uncertainty quantification for both predictions and sensitivity-index estimates.

Some metamodels struggle more than others with the two stages in the above approach, namely fitting the metamodel, then using the fitted metamodel to estimate the sensitivity indices.
Regarding the first stage, many of these metamodel-based approaches struggle to fit if the number of inputs $p$ and function evaluations $n$ are not small. 
A GP has $O(n^3)$ computation time and struggles to fit for even $p=10$. 
PCE has been fit for $p=25$, but it has been noted that PCE struggles to fit for larger $p$ \citep{Sudret2008,Crestaux09}.
BMARS works for $p=200$ for Sobol' indices \citep{francom2018sensitivity}.
Figure~\ref{fig:simstudy500} of this paper provides an example where BART fits to a $p=500$ scenario with $d=250$ active variables.
Regarding the second stage, if a metamodel is cheap to evaluate, then the fitted metamodel's Shapley effects can be estimated using Monte Carlo integration of the Shapley-effect integrals, as done in Algorithm~1 from \cite{song2016shapley} or a parallelized version of it \citep{zhang2024variable}. However, this will create another layer of approximation error that can be avoided if the metamodel allows for exact computation. On this front,
BART \citep{horiguchi2021assessing}, BMARS \citep{francom2018sensitivity}, and PCE \citep{Sudret2008} have closed-form expressions for Sobol\'{} indices (and thus for Shapley effects) that can be computed exactly once the metamodel is fit. 
Such expressions also exist for GPs with polynomial mean and either a separable Gaussian, Bohman, or cubic correlation function \citep{Oakley2004,Chen2005,Chen2006,Marrel2009,Moon2010,svenson2014estimating,Santner18}. 
Table~1 in the Supplementary Material summarizes these metamodel properties.

To our knowledge, this article is the first to provide an estimator of a function's \textit{Shapley effects} that is computationally tractable for a relatively large number of inputs and function evaluations, as well as provide theoretical guarantees of consistency in the context of nonparametric regression where the function is observed with noise.  
BART approximates a function by a piecewise-constant function whose exact Sobol\'{} indices are provided by  \cite{horiguchi2021assessing} and can be easily computed (we will refer to these as ``BART-based Sobol\'{} indices'' for the rest of this article).
Section~\ref{sec:prelim} will show these closed-form expressions can also be used to compute BART-based Shapley effects, but because the number of expressions to compute increases dramatically, Section~\ref{sec:computation} discusses computationally friendly approximations. 
On the other hand, our contraction-rate results rely heavily on recent BART theory from \cite{jeong2023art}, who introduce the large class of sparse piecewise heterogeneous anisotropic Hölder functions and show that over this function class, the contraction rate for Bayesian forests is optimal up to a logarithmic factor. 

This article is organized as follows. 
Section~\ref{sec:prelim} reviews BART, Sobol\'{} indices, and Shapley effects. 
Section~\ref{sec:computation} provides our main theoretical posterior contraction results and discusses the computation of BART-based Shapley effects.  Section~\ref{sec:numerical} showcases their performance on numerical examples, including data from the En-Roads climate simulator (analogous discussion for BART-based Sobol\'{} indices can be found in \cite{horiguchi2021assessing}).
Section~\ref{sec:discussion} provides discussion on future work. 
Our results on posterior contraction for BART-based Sobol\'{} indices and Shapley effects, as well as proofs of these results, are included as Supplementary Material.

\section{Review}
\label{sec:prelim}

Mirroring \cite{jeong2023art}, this article considers regression settings with either a fixed or random design. 
The regression model with \textit{fixed} design is
\begin{align}
\label{eq:fixeddesign}
Y_i = f_0(\x_i) + \varepsilon_i, \qquad \varepsilon_i \sim N(0, \sigma^2_0), \qquad i = 1, \ldots, n,
\end{align}
where $\sigma^2_0 < \infty$ and each covariate $\x_i \in [0,1]^p$ is fixed.  
A fixed design would be assumed if, for example, the trees in BART are allowed to split only on observed covariate values (which was a specification used in the seminal BART paper \citep{chipman2010bart}) or on dyadic midpoints of the domain. 
The regression model with \textit{random} design is
\begin{align}
\label{eq:randomdesign}
Y_i = f_0(\X_i) + \varepsilon_i, \qquad \X_i \sim \pi, \qquad \varepsilon_i \sim N(0, \sigma^2_0), \qquad i = 1, \ldots, n,
\end{align}
where $\sigma^2_0 < \infty$, each $\X_i \in [0,1]^p$ is a $p$-dimensional random covariate, and $\pi$ is a probability measure such that $\text{supp}(\pi) \subseteq [0,1]^p$.
A random design would be assumed for estimation problems such as density estimation or regression/classification with random design.
Our posterior contraction results deal separately with fixed or random designs. 

\subsection{BART}
\label{sec:bartreview}

\begin{figure}[h!]
\centering
    \begin{tikzpicture}[->,>=stealth', level/.style={sibling distance = 4cm/#1, level distance = 1.7cm}, scale=0.8, transform shape]
        \node [treenode, color=red] {$x_2 < 0.7$}
        child {
            node [treenode, color=blue] {$x_1 < 0.2$} 
            child { node [terminal] {$\mu_1$} }
            child { node [terminal] {$\mu_2$} }
        }
        child {
            node [treenode, color=green] {$x_1 < 0.4$}      
            child { node [terminal] {$\mu_3$} }
            child { node [terminal] {$\mu_4$} }
        }
    ;
    \end{tikzpicture}
    ~
    \begin{tikzpicture}[scale=0.65]
        \begin{axis}[
        xmin=0, xmax=1, xlabel={$x_1$}, 
        ymin=0, ymax=1, ylabel={$x_2$}]
            \addplot [dashed, thick, color=red] table {
                0 0.7
                1 0.7
            };
            \addplot [dashed, thick, color=blue] table {
                0.2 0
                0.2 0.7
            };
            \addplot [dashed, thick, color=green] table {
                0.4 0.7
                0.4 1
            };
            \node[scale=1.5] at (axis cs:0.1,0.35) {$\mu_1$};
            \node[scale=1.5] at (axis cs:0.6,0.35) {$\mu_2$};
            \node[scale=1.5] at (axis cs:0.2,0.85) {$\mu_3$};
            \node[scale=1.5] at (axis cs:0.7,0.85) {$\mu_4$};
        \end{axis}
    \end{tikzpicture}
\caption{An example tree shown graphically (left) and as a piecewise-constant regression function (right) on the input space $[0,1]^2$.}
\label{fig:tree1}
\end{figure}

In a regression setting in the form of either \eqref{eq:fixeddesign} and \eqref{eq:randomdesign}, a BART model approximates the unknown function $f_0$ by a sum of $T$ regression trees:
\begin{equation}
f_0(\cdot) \approx \sum_{t=1}^T g(\cdot; \Theta_t),
\label{eq:modelSOT}
\end{equation}
where each regression-tree function $g(\bm{\cdot}; \Theta_t)\colon [0,1]^p \rightarrow \mathbb{R}$ is piecewise constant over the input space.
Each parameter set $\Theta_t$ determines a partition of the input space $[0,1]^p$ into boxes (i.e. hyperrectangles) and the fitted response values assigned to each partition piece. 
The partition is induced by recursively applying binary splitting rules;
Figure \ref{fig:tree1} shows an illustrative example. 
To regularize the model fit, the BART prior over the parameters $\{\Theta_t\}_{t=1}^T$ keeps the individual tree effects small, which causes each function $g(\bm{\cdot}; \Theta_t)$ to contribute a small portion to the total approximation of $f_0$. 
The expected response $\E\big[Y(\x) \mid \{\Theta_t\}_{t=1}^T \big]$ at a given input $\x$ is then the sum of each contribution $g(\x; \Theta_t)$.

Though the right hand side of \eqref{eq:modelSOT} is piecewise constant, \cite{jeong2023art} shows that under certain conditions, BART can approximate the unknown function $f_0$ (which itself need not be piecewise constant) arbitrarily closely with attractive posterior contraction rates.
For space consideration, the Supplementary Materials will describe the types of functions that BART can capture and the conditions made in the theorems of \cite{jeong2023art} that our contraction-rate results rely on.

\subsection{Sobol\'{} indices}
\label{sec:sobol}

Let $L^2 \equiv L^2([0,1]^p)$ denote the space of real-valued, square-integrable functions on the hypercube $[0,1]^p$.
\cite{sobol1990sensitivity,Sobol93} shows that if the random variable $\X$ follows an orthogonal distribution whose support is $[0,1]^p$ and if $f \in L^2$, 
then the variance of $f(\X)$ can be decomposed into a sum of terms attributed to single inputs or to interactions between sets of inputs:
\begin{align}
\label{eq:variancedecomposition}
    \V f(\X) = \sum_{j=1}^p V_j + \sum_{j=1}^p \sum_{k<j} V_{jk} + \cdots + V_{1,2,\ldots,p}
\end{align}
where we recursively define for each variable index set $P \subseteq [p]$
\begin{align*}
    V_P \coloneqq \V [\E \{f(\X) \mid \X_P\}] - \sum_{Q \subset P} V_Q
\end{align*}
where we set $V_{\emptyset} = 0$ and the relation $\subset$ denotes a strict subset.
For any variable index $j \in [p]$, the term $V_{\{j\}} = V_j$ is known as the $j$th (unnormalized) first-order (or main-effect) Sobol\'{} index, 
and the sum $T_j = \sum_{P \subseteq ([p] \setminus \{j\})} V_{P \cup \{j\}}$ is known as the $j$th (unnormalized) total-effect Sobol\'{} index.
We note that $T_j \geq V_j \geq 0$ for all $j \in [p]$.

The $V_P$ terms in \eqref{eq:variancedecomposition} are often divided by the total variance to produce the normalized terms $V_P / [\V f(\X)]$, 
which have the nice interpretation of being the proportion of the total variance attributed to the interaction between the variables whose indices are in the index set $P$.
If $P$ is the singleton $\{j\}$, then the normalized term $V_j / [\V f(\X)]$ can be interpreted as the proportion of the total variance attributed to variable $j$ by itself.
Despite this nice interpretation, the remainder of the article will assume that such indices are unnormalized unless otherwise stated.

To see why these indices' interpretation requires $\X$ to follow an orthogonal distribution, 
we extend the definition of $V_P$ by allowing $\X$ to follow a possibly non-orthogonal distribution $\pi$ whose support is $[0,1]^p$. 
We first define the functional $c_{P,\pi}\colon L^2 \rightarrow \R$ as 
\begin{equation}
\label{eq:c}
c_{P,\pi}(f) = \V_{\pi} [\E_{\pi} \{f(\X) \mid \X_P\}] 
\end{equation}
for any $f \in L^2$.
Then the generalized $V_P$ under the distribution $\pi$ is recursively defined as \[V_{P, \pi}(f) \coloneqq c_{P,\pi}(f) - \sum_{Q \subset P} V_{Q,\pi}(f),\] where again we set $V_{\emptyset,\pi}(f) = 0$. 
Similarly, we define the generalized $j$th total-effect term: \[T_{j, \pi}(f) = \sum_{P \subseteq ([p] \setminus \{j\})} V_{P \cup \{j\}, \pi}(f)\]
where $\subseteq$ denotes a subset that is not necessarily strict.
Recall that if $\pi$ is orthogonal and $f \in L^2$, then $T_{j, \pi}(f) \geq V_{j, \pi}(f) \geq 0$ for all $j \in [p]$ and the variance decomposition \eqref{eq:variancedecomposition} (where orthogonality implies $V_P = V_{P, \pi}(f)$ for all $P \subseteq [p]$) holds.
However, Theorem 2 of \cite{song2016shapley} asserts the existence of a non-orthogonal distribution $\pi$ and a function $f \in L^2$ such that $\sum_{j=1}^p V_{j,\pi}(f) > \V_{\pi} f(\X) > \sum_{j=1}^p T_{j,\pi}(f)$. 
In such a case, these Sobol\'{} indices can no longer be interpreted as in the orthogonal case.

\subsection{Shapley effects}
\label{sec:shapley}

One way to measure variable activity, regardless of dependence among inputs, are the Shapley effects defined by \cite{song2016shapley} as the Shapley values in \cite{owen2014sobol} using the functional \eqref{eq:c} as the ``value'' or ``cost.''
For $j \in [p]$ the $j$th Shapley effect is defined as
\begin{equation} \label{eq:shapley}
    S_{j, \pi}(f) = (p!)^{-1} \sum_{P \subseteq ([p] \setminus \{j\})} (p - |P| - 1)! \, |P|! \, \big\{c_{P \cup \{j\},\pi}(f) - c_{P,\pi}(f)\big\},
\end{equation}
which has the desirable property $\sum_{j=1}^p S_{j, \pi}(f) = \V_{\pi} f(\X)$
for any distribution $\pi$ (possibly nonorthogonal) whose support is $[0,1]^p$. 
Hence, the $j$th (normalized) Shapley effect can be nicely interpreted as the contribution of input $j$ to the total output variance.
Furthermore, if $\pi$ is orthogonal, then 
\begin{equation}
    V_{j, \pi}(f) \leq S_{j, \pi}(f) \leq T_{j, \pi}(f)
    \label{eq:sobolshapley}
\end{equation}
for any $f \in L^2$ and $j \in [p]$ \citep[Section 3]{owen2014sobol},
i.e. the $j$th Shapley effect is bounded between the $j$th main-effect and total-effect Sobol\'{} index. 

Calculating \eqref{eq:shapley} can be prohibitively costly due to it being a sum of values \eqref{eq:c} over all subsets of a set $[p]\setminus \{j\}$.
Its computational tractability will be discussed in Section~\ref{sec:computation}.

\section{Main results and computation of Shapley effects}
\label{sec:computation}

This section will address theoretical support and computation of Shapley effects using a BART metamodel. The metamodel-based approach in estimating Shapley effects has two approximation layers: how well the metamodel approximates (functionals of) the underlying regression function $f_0$, and how well the Shapley-effect estimates approximate the Shapley effects of the metamodel function.

\subsection{Consistency Result}
For the second layer, we establish posterior consistency for our BART-based Shapley effects using (first-layer) posterior consistency for BART from \cite{jeong2023art}.
The required theoretical results, fully developed in the Supplementary Material, characterize the posterior contraction as the dataset size $n\rightarrow \infty$. The contraction rate quantifies how quickly the posterior distribution approaches the underlying function's true Shapley effects. In particular, for random designs, we have the following. 
\begin{corollary}
\label{corr:randomdesign}
Under the assumptions of Theorem 4 of \cite{jeong2023art} -- Assumptions (A1), (A2), (A3$\ast$), (A4), (A5), (A6$\ast$), and (A7), and the prior assigned through (P1), (P2$\ast$), and (P3$\ast$) -- and Theorem~3 in Section~S7, there exist positive constants $L_{V,\pi,|P|}$, $L_{T,\pi}$, and $L_S$ such that 
as $n \rightarrow \infty$ for $\epsilon_n$ in Eq.\ (S5.4) in Section~S7, 
\begin{align*}
\E_0 \Pi \Big\{(f, \sigma^2)\colon |V_{P, \pi}(f) - V_{P, \pi}(f_0)| + |\sigma^2 - \sigma^2_0| > L_{V,\pi,|P|} \epsilon_n  \Big| Y_1, \ldots, Y_n \Big\} &\rightarrow 0, \\
\E_0 \Pi \Big\{(f, \sigma^2)\colon |T_{j, \pi}(f) - T_{j, \pi}(f_0)| + |\sigma^2 - \sigma^2_0| > L_{T,\pi} \epsilon_n  \Big| Y_1, \ldots, Y_n \Big\} &\rightarrow 0, \\ 
\text{ and } \E_0 \Pi \Big\{(f, \sigma^2)\colon |S_{j, \pi}(f) - S_{j, \pi}(f_0)| + |\sigma^2 - \sigma^2_0| > L_S \epsilon_n  \Big| Y_1, \ldots, Y_n \Big\} &\rightarrow 0.
\end{align*}
\end{corollary}
The supplement contains a similar result for fixed designs, as well as proofs for all theoretical results.

\subsection{Shapley Effect Computation}

The remainder of this section will address the computation of Shapley effects, and how well the Shapley-effect estimates approximate the Shapley effects of the metamodel function. Since BART is a Bayesian metamodel, our focus is to address the computational aspects when $f_0$ is approximated by $n_{draw}$ posterior draws $\hat{f}^{(1)}, \ldots, \hat{f}^{(n_{draw})}$ of the fitted metamodel.

For each input $j \in [p]$, we can construct a posterior distribution for the $j$th Shapley effect $S_{j,\pi}(f_0)$ of $f_0$ using the $n_{draw}$ values 
\begin{align} \label{eq:shapleydrawexact}
    S_{j,\pi}(\hat{f}^{(i)}) 
    &=  \sum_{P \subseteq ([p] \setminus \{j\})} \frac{(p - |P| - 1)! \, |P|!}{p!}  \Big[c_{P \cup \{j\},\pi}(\hat{f}^{(i)}) - c_{P,\pi}(\hat{f}^{(i)})\Big],
\end{align}
for $i=1,\ldots,n_{draw}$.
We can use the sample mean of $\{S_{j,\pi}(\hat{f}^{(i)})\}_{i=1}^{n_{draw}}$ as a point estimate for $S_{j,\pi}(f_0)$, and use the end points of the middle $95\%$ values as a $95\%$ credible interval for $S_{j,\pi}(f_0)$.
For each $\hat{f}^{(i)}$, computing its $p$ Shapley effects would require computing the cost function \eqref{eq:c} for $2^p$ subsets of the set $[p]$.
The exponential increase in $p$ is undesirable, but also the calculation of even a single cost function might be computationally intractable if $p$ is large enough. 

\subsection{Sum over subsets}
\label{sec:sumoversubsets}

We first tackle the exponential increase in $p$. 
If the inputs are orthogonal, then \eqref{eq:shapleydrawexact} is bounded between the main-effect and total-effect Sobol\'{} indices, so we can avoid computing \eqref{eq:shapleydrawexact} entirely by crudely estimating it with, for example, the mean of the two Sobol\'{} indices.
If the inputs are not assumed to be orthogonal, we can reduce the increase from exponential to linear by using the following random-subset approach. 
Rather than compute the cost difference in \eqref{eq:shapleydrawexact} for all $2^{p-1}$ subsets, we instead compute the cost difference for only a small number $m$ of subsets that are randomly created by including each $j' \in ([p]\setminus \{j\})$ with probability $0.5$. 
(We could incorporate prior information about which inputs $j'$ are important by increasing or decreasing the probability of including any particular $j'$. We could also incorporate prior information about interactions between groups of inputs by increasing the probability that they appear together in a subset and decreasing the probability that they appear separately. We save further exploration for future work.)
Hence any subset $P \subseteq ([p]\setminus \{j\})$ is chosen with probability $|P|!(p-(|P|+1))! / (p!)$.
Under this approach, we approximate \eqref{eq:shapleydrawexact} by
\begin{align}
    S_{j,\pi}(\hat{f}^{(i)})  \approx m^{-1} \sum_{l=1}^m \Big\{c_{P_l^{(i)} \cup \{j\},\pi}(\hat{f}^{(i)}) - c_{P_l^{(i)},\pi}(\hat{f}^{(i)})\Big\}, 
    \label{eq:shapleydrawrandom}
\end{align}
where $P_l^{(i)}$ is the $l$th of $m$ randomly drawn subsets of $([p]\setminus\{j\})$ for the $i$th posterior draw.
Hence this approach reduces the number of cost-function calculations from $n_{draw} \times 2^p$ to $n_{draw} \times p \times (2m)$.

(As an aside, this approach is equivalent to how subsets are chosen in the random-permutation scheme of \cite{castro2009polynomial}: under this scheme, any subset $P \subseteq ([p]\setminus\{j\})$ can be obtained by any permutation of $[p]$ whose first $|P|$ elements are the elements in $P$ and whose $(|P|+1)$th element is $j$. Because there are $|P|!(p-(|P|+1))!$ such permutations of $[p]$ and each permutation of $[p]$ is drawn with equal probability $(p!)^{-1}$, the subset $P$ is selected with probability $|P|!(p-(|P|+1))! / (p!)$. \cite{song2016shapley,van2018new,yang2024fast} provide improvements and modifications on this ``simple random sampling'' of permutations.)

What value of $m$ should be used for \eqref{eq:shapleydrawrandom}?
Both its computational cost and its fidelity to \eqref{eq:shapleydrawexact} increase with $m \times n_{draw}$. We argue that if $n_{draw} \gg 1$ (a requirement for any decent posterior summary of the surrogate model), then $m=1$ random subset will suffice. Because the randomness from the original MCMC mechanism is independent of how the random subsets are chosen, \eqref{eq:shapleydrawrandom} is just a noisier version of \eqref{eq:shapleydrawexact}.
Our approach does not inflate the correlation between any two MCMC draws.
Also, it does not introduce any additional bias, seeing as each random subset $P$ is i.i.d.\ and is drawn with probability exactly equal to the weight $(p-|P|-1)! |P|! / (p!)$ in the Shapley-effect expression \eqref{eq:shapley}. 
Now we consider the additional variability induced by the random subsets.
For each $j\in[p]$, we aim to approximate the cumulative distribution function (CDF) $F_{n,j}$ of the posterior distribution of the $j$th Shapley effect by an empirical CDF comprised of $n_{draw}$ MCMC draws. 
(Here the sample size $n$ is fixed and finite, but if the metamodel has posterior consistency, then $\lim_{n\rightarrow \infty} F_{n,j}$ would be a step function consisting of a single jump located at the true $j$th Shapley effect of the underlying regression function $f_0$.) 
Generally speaking, MCMC draws are thinned in order to be approximately i.i.d.\ from the limiting distribution. 
If we assume i.i.d.\ MCMC draws, Donsker's theorem \citep{donsker1952justification} tells us that as $n_{draw} \rightarrow \infty$, both empirical CDFs --- one using random subsets, the other using exact subsets --- will converge to the target CDF $F_{n,j}$ at the same rate $\mathcal{O}(n_{draw}^{-1/2})$, regardless of how many random subsets are used. 
Thus, we do not need to increase the number of random subsets in order to reduce the additional variability, since this variability will already shrink to zero as $n_{draw} \rightarrow \infty$. 
Hence we use $m=1$ for all experiments in this paper.

\subsection{Cost calculation}
\label{sec:costcalculation}

We now consider how each calculation of \eqref{eq:c} is affected by which metamodel is used.
For any metamodel that can be evaluated cheaply, Algorithm~1 of \cite{song2016shapley} can be used to approximate the integrals in \eqref{eq:c} for the metamodel by first sampling from the input distribution many times and then evaluating the metamodel on the many generated inputs. However, keeping the resulting integral approximation error small will likely require the number of random permutations and hence the computation time for each integral to grow at least linearly in $p$ \citep{tang2024note}, and thus for all $p$ inputs the computation time to grow at least quadratically in $p$.
Additionally, approximating \eqref{eq:c} in this way produces two inference issues illustrated in the independent-inputs example in Section~\ref{sec:numerical5}. First, even though the exact cost difference in \eqref{eq:shapley} is nonnegative for independent inputs by definition, the \textit{estimated} cost difference can be negative and hence often produces negative Shapley-effect values for inert inputs (see e.g., the Morris function GP estimates in Section~\ref{sec:numerical5}), which creates interpretability issues.
Second, the variability due to approximating \eqref{eq:c} is much larger than the variability due to the random subsets/permutations (see Section~\ref{sec:numerical5} for more detail on this point). 

For these reasons, it is desirable to compute \eqref{eq:c} exactly, which can be done for some metamodels.
For BART, a closed-form expression for \eqref{eq:c} can be found using Theorem~1 of \cite{horiguchi2021assessing}.
For Bayesian MARS, \cite{francom2018sensitivity} provides a closed-form expression for estimating Sobol\'{} indices and contains a numerical example with $p=200$.
For a GP, a closed-form expression can be found for certain correlation functions, but these functions are typically restrictive i.e., assume stationarity and isotropy.
Exact computation of each of these expressions is easy when the inputs are independent, but otherwise is challenging, in which case we will resort to using Algorithm~1 of \cite{song2016shapley}.

\section{Numerical examples}
\label{sec:numerical}

This section explores the numerical performance of BART-based Shapley effects. 
(BART-based Sobol\'{} indices are evaluated in detail in \cite{horiguchi2021assessing} and \cite{horiguchi2020bayesian} and hence are not evaluated in this paper.)
Details of the test functions used in this section are listed in Section~S6 in the Supplementary Material.

\subsection{Exact vs Monte Carlo cost calculation}
\label{sec:numerical5}

This section explores the computational and accuracy differences between computing the cost exactly as in \eqref{eq:shapleydrawexact} and estimating the cost using Algorithm~1 of \cite{song2016shapley} when the inputs are independent.
For our first set of experiments, we create a dataset with $n=50p$ observations and noise variance $\sigma^2_0 = 0.25 Var(f(\X))$ from \eqref{eq:fixeddesign} for each test function $f$ and each $p \in \{5,50,200\}$. 
To each dataset, we fit a BART model with $n_{draw}=1000$ posterior draws and $200$ trees with code from \cite{openbt}. 
For comparison, we also fit a Gaussian process (GP) model and estimate the Shapley effects of the fitted GP mean model using Algorithm~1 of \cite{song2016shapley} as implemented in \texttt{shapleyPermRand}, which was the only function in the \texttt{sensitivity} R package \citep{sensitivityRpackage} that we found could fit to our $p=50$ data sets in a reasonable amount of time. 
Parameter specifications are in the caption of Figure~\ref{fig:simstudygp}.
For the $p=200$ cases, we could not read the large GP-model file sizes into \text{R} and hence do not include these results. 

\begin{figure}[h!]
    \centering
    \includegraphics[width=\textwidth]{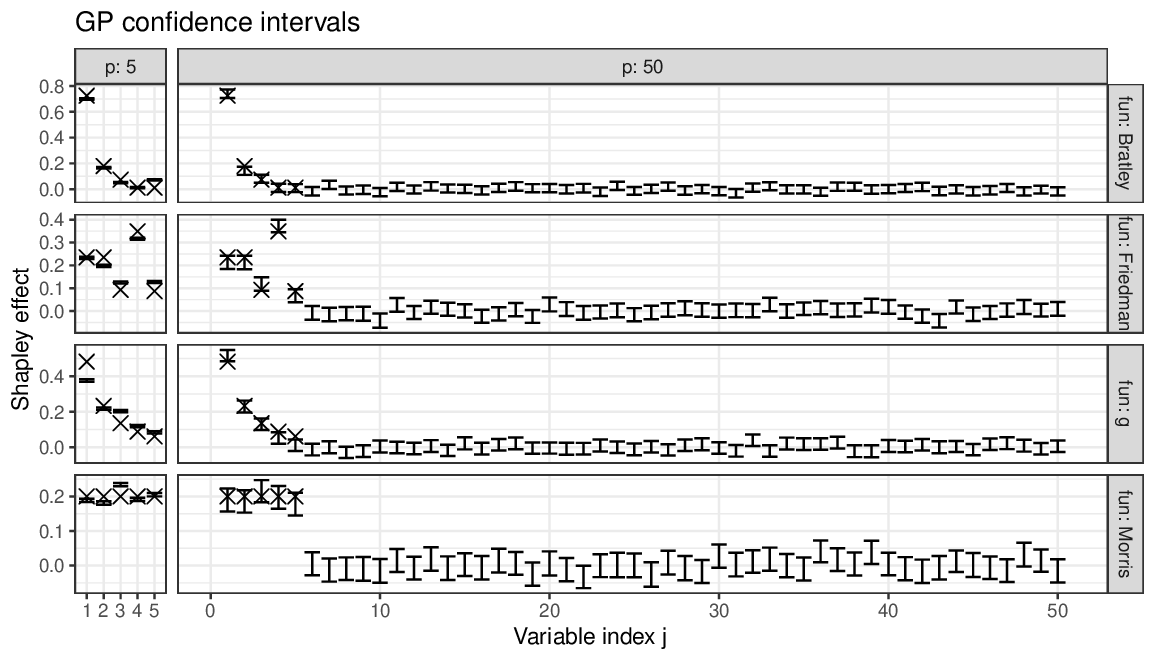}
    \caption{$95\%$ confidence intervals (computed by \texttt{sensitivity}) for the Shapley-effect estimates from a GP fit to $n=50p$ observations with $d=5$ active variables. Crosses indicate a function's true Shapley effects. The \texttt{shapleyPermRand} approach samples $N_V + m(p-1) N_O N_I$ inputs to estimate expectation. 
    Per \cite{song2016shapley}, we set $N_V=10^5$ samples to estimate the total variance, and $N_O=1$ and $N_I=3$ to estimate the outer and inner expectations, respectively.
    For $p=5$ we use $m=10^5$ random permutations; for $p=50$ we reduce this to $m=3\times10^3$ to avoid numerical overflow.
    }
    \label{fig:simstudygp}
\end{figure}

\begin{figure}[h]
    \centering
    \includegraphics[width=\textwidth]{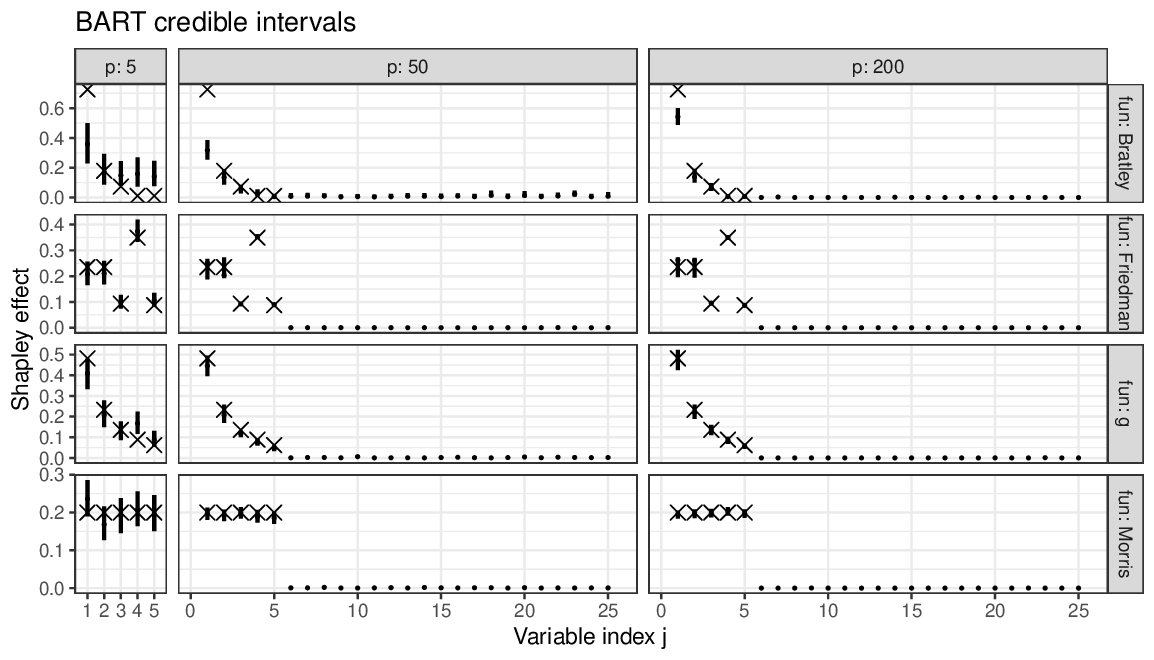}
    \caption{$95\%$ credible intervals (as in \eqref{eq:shapleydrawrandom}) over 1000 posterior draws for the Shapley-effect estimates from a BART model fit to $n=50p$ observations with $d=5$ active variables. Crosses indicate a function's true Shapley effects. 
    For $p=50$ and $p=200$, only the first 25 input variables are shown for space considerations.
    }
    \label{fig:simstudy}
\end{figure}

We first compare the GP estimates to the BART estimates. 
Figure~\ref{fig:simstudygp} shows the Shapley-effect confidence intervals of the GP approach as computed by the \texttt{sensitivity} package, and 
Figure~\ref{fig:simstudy} shows our Shapley-effect credible intervals using BART as defined in \eqref{eq:shapleydrawrandom}. 
The GP model seems to capture the large Shapley effects better than the BART model does, which might be explained by the fact that the data-generating functions are all continuously differentiable and thus are well suited for GPs.  
However, the GP model also seems to have more trouble setting the inactive variables to have zero estimated Shapley effect; indeed, for the $g$-function with $p=50$, the confidence intervals for many of the inactive variables are higher than the interval for the active variable $j=5$.
Furthermore, the GP confidence intervals for all {\em inactive} variables cover negative values (as computed by the \texttt{sensitivity} package), even though Shapley effects are nonnegative by definition.  In contrast, the BART credible intervals never cover negative values.
Finally, the GP confidence intervals are wider than the BART credible intervals for the inert inputs, which is surprising since both intervals capture the variability due to the Shapley-effect estimation of the metamodel, but the BART intervals also capture the metamodel's posterior uncertainty (i.e., how accurately the metamodel fits the true regression function), whereas the GP intervals do not. (We emphasize that this stems from the different methods of estimating the Shapley effects of the respective metamodels --- see Section~\ref{sec:costcalculation} --- rather than from the difference between GP and BART.) Hence we can conclude that in this scenario, the $m=1$ subset approximation in Section~\ref{sec:sumoversubsets} produces negligible errors, especially compared to the approximation error from the estimation of \eqref{eq:c} even with $3000$ random permutations.

Theoretically, we could increase the number of random permutations in order to reduce the Monte Carlo error of estimating \eqref{eq:c} to an arbitrarily small amount, but as discussed in Section~\ref{sec:costcalculation}, maintaining a small approximation error for computing the cost function \eqref{eq:c} for all $p$ inputs will likely require the computation time to grow at least quadratically in $p$.
But how does the exact computation of \eqref{eq:c}, which we recall is currently only implemented for BART, scale with increasing $p$?
For two regression functions, $p=3,\ldots,10$, and fixed $n=500$, 
Figure~\ref{fig:p_runtime_n500} shows that the computation time of training a BART model appears to be constant in $p$.
The figure also shows that computing the Shapley effects of the fitted BART model using the exact \eqref{eq:c} appears to grow faster than quadratic in $p$. 
This implies that the exact-cost approach scales better in $p$ than the Monte Carlo approach if the number of random permutations increases in order to maintain a small Monte Carlo approximation error from estimating \eqref{eq:c}.
(Because the computation times for the exact-cost approach is specific to BART, the behavior may not generalize if the exact-cost approach is implemented for other metamodels.)

\begin{figure}[ht]
    \centering
    \includegraphics[width=\textwidth]{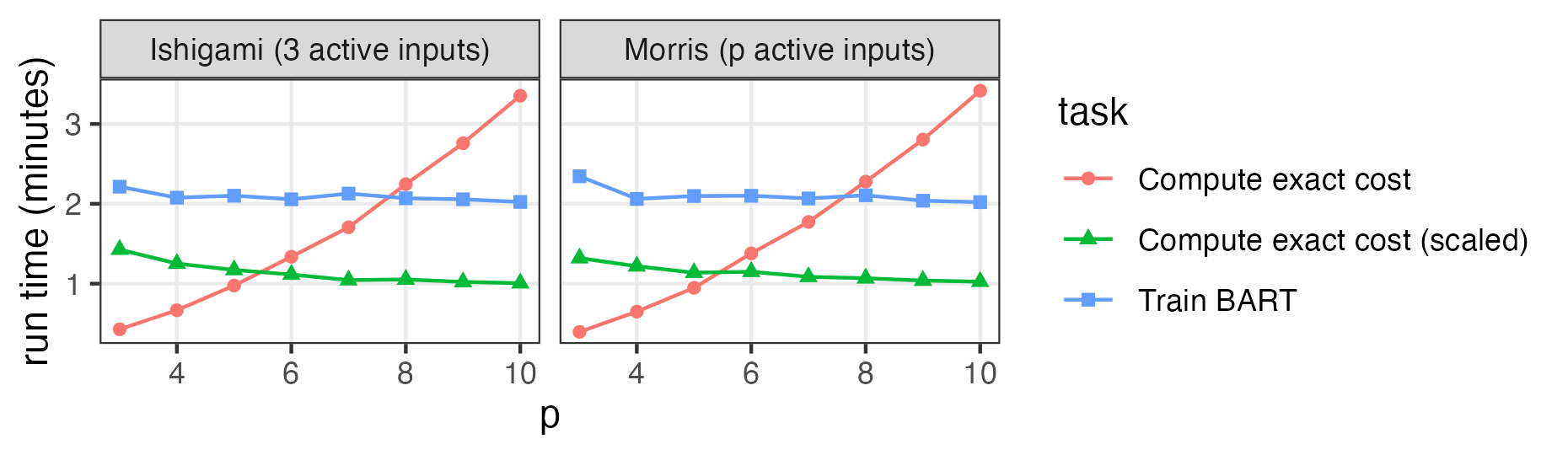}
    \caption{For $n=500$ and various $p$, the figure shows the run times of training the BART model (1000 posterior draws, 9000 burn-in draws, M1-chip 4-core laptop) and of computing the Shapley effects of the fitted BART model using the exact \eqref{eq:c}. ``Compute exact cost (scaled)'' indicates the latter run time divided by $p^2$ (and multiplied by 30 for better visual comparison) to get an upper bound on how it scales with $p$ for fixed $n$. 
    }
    \label{fig:p_runtime_n500}
\end{figure}

Similarly, Figure~\ref{fig:p_runtime_p3} explores how these calculations scale with increasing $n$ for fixed $p=3$.
The BART training time appears to grow sublinearly in $n$ (a little faster than rate $\sqrt{n}$), and the Shapley-effect calculation time using exact computation of \eqref{eq:c} appears to grow more slowly than $\log(\log(n))$.
This justifies the use of BART if $n$ increases with $p$.
(Section~5 of \cite{pratola2014parallel} studies the scalability of a parallel BART MCMC algorithm.)

\begin{figure}[ht]
    \centering
    \includegraphics[width=\textwidth]{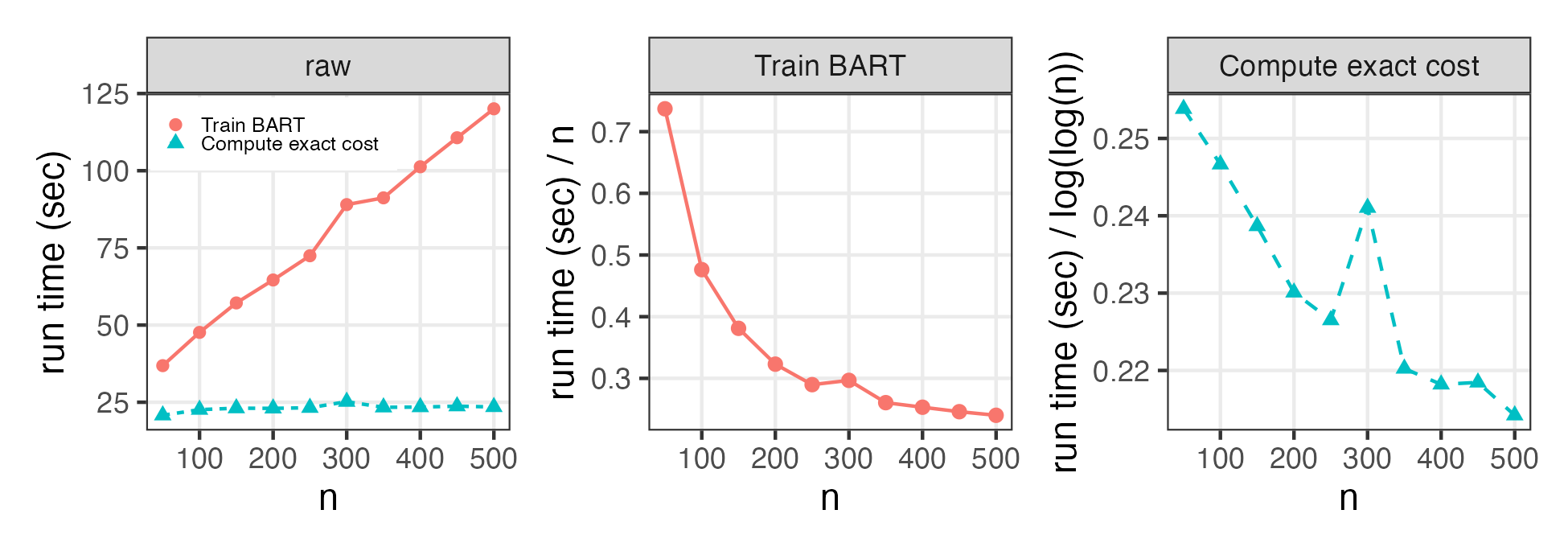}
    \caption{Left: run times of training a BART model (1000 posterior draws, 9000 burn-in draws, M1-chip 4-core laptop) and of computing the Shapley effects of the fitted BART model using the exact \eqref{eq:c} for $p=3$ and various $n$. Center and right panels scale these run times to get an upper bound on how the two run times scale with $n$ for fixed $p=3$. 
    }
    \label{fig:p_runtime_p3}
\end{figure}

We further examine the BART results in Figure~\ref{fig:simstudy}. 
For the Friedman and Morris functions, the true Shapley effects are contained in the credible intervals and are often near the center of the intervals. 
For the $g$-function, the $p=5$ scenario shows the credible intervals struggling a bit to capture the true Shapley effects, but the $p=200$ scenarios show better performance from the intervals.
This $p=200$ result becomes even more notable if we consider the fitted BART models do not use (P1)'s tree prior with Dirichlet sparsity from \cite{Linero18}, and that the $g$ function is purely a product of univariate functions. 
For the Bratley function, the intervals struggle quite a bit to capture the true Shapley effects.
For this challenging Bratley function, we next explore what parameters or priors should be changed to improve the Shapley-effect estimates.
Of the three directions we explored -- increasing the number of trees to $300$, weakening the tree-depth prior to encourage higher order interactions, and increasing $n$ --
only the third (with $200$ trees, the same tree-depth prior as in the first set of explorations, and $p=5$) yielded estimates closer to the true Shapley effects. 
This provides assurance that for these more challenging functions, the estimates can be close to the true Shapley effects if $n$ is large enough without having to change any other parameters or priors.

\begin{figure}[h]
\centering
\includegraphics[width=\textwidth]{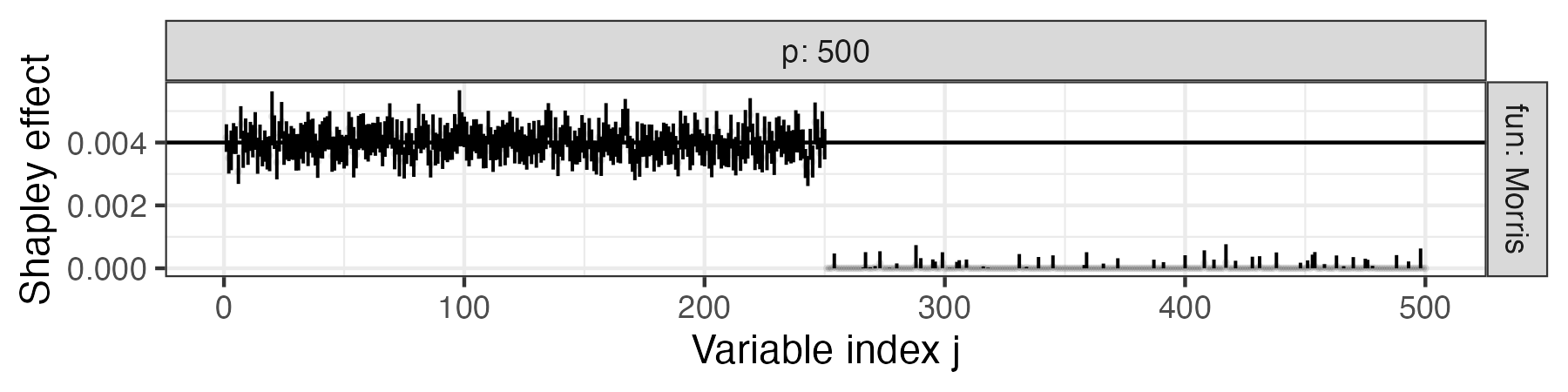}
\caption{$95\%$ credible intervals (as computed in \eqref{eq:shapleydrawrandom}) over 300 posterior draws for the Shapley-effect estimates from a BART model fit to $n=50p$ observations. Horizontal line corresponds to the function's true Shapley effects of the $d=250$ active variables.}
\label{fig:simstudy500}
\end{figure}

Finally, we explore the performance of BART-based Shapley effects for $d=250$ active input variables and input dimension $p=500$, which is a regime that bottlenecks most other methods. 
(We omit GP results here due to not being able to compute GP-based Shapley-effect estimates.)
For ease, we use the Morris function since its Shapley effects are $1/d$ for the $d$ active variables.
Figure~\ref{fig:simstudy500} shows that BART clearly distinguishes between the first 250 inputs (these intervals are centered around $1/d$) and the second 250 inputs (these intervals are centered around zero) for such a large $p$. 

\subsection{Robustness of Shapley-effect estimates to input correlation}
\label{sec:robustness}

Now we explore how sensitive the estimated Shapley effects are to varying degrees of correlation.
When the inputs are not independent, our proposed BART-based estimator requires being able to compute the input probability measure of various subsets of the input space $[0,1]^p$ (this point is discussed further in Section~\ref{sec:discussion}).
Hence, here we compute the cost function \eqref{eq:c} using Monte Carlo integration as implemented by the \texttt{sensitivity} R package.

\begin{figure}[h]
\centering
\includegraphics[width=\textwidth]{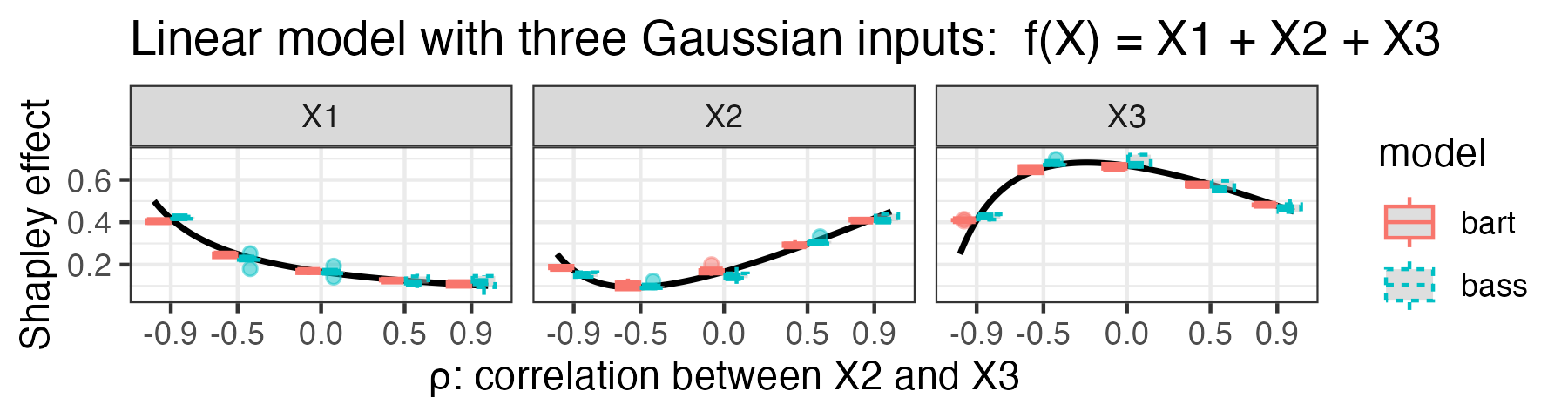}
\caption{Shapley effects estimated by BART and BMARS \citep[as implemented in the \texttt{BASS} R package,][]{bass2020}. The boxplot variability comes from 5 replicates of the process described in Section~\ref{sec:robustness}. The solid curves represent the function's true Shapley effects.}
\label{fig:inputcorr}
\end{figure}

It can be difficult to analytically compute Shapley effects for even moderately complicated functions. 
Here we use the regression function $f(\mathbf{x}) = x_1 + x_2 + x_3$, where the inputs have a trivariate Gaussian distribution with mean zero and covariance matrix $\Sigma$ with entries $\Sigma_{11} = \Sigma_{22} = 1$, $\Sigma_{33} = 4$, $\Sigma_{23} = \Sigma_{32} = 2 \rho$ (where $-1 < \rho < 1$), and zero for all other entries. 
Thus, the input distribution is parameterized by the correlation $\rho$ between input variables $X_2$ and $X_3$.
For each $\rho \in \{-0.9, -0.5, 0, 0.5, 0.9\}$, we generate $n=1000$ input values according to $\rho$, and then evaluate the regression function with additive Gaussian noise whose standard deviation is $0.1$ times the standard deviation of the regression function under $\rho$.
We then fit a BART model and a Bayesian MARS (BMARS) model to these observations before using the \texttt{shapleyPermEx()} function from the \texttt{sensitivity} R package to estimate the function's Shapley effects. (The parameters we use are \texttt{Nv=1000, No=100, Ni=3}.)
We repeat this process five times.

Figure~\ref{fig:inputcorr} shows the estimated Shapley effects for these five correlation values. 
We see that both BART and BMARS seem to recover the true Shapley effects even for large correlations between inputs 2 and 3.
There also does not seem to be any systematic performance difference between the two metamodels.
Hence, if we extrapolate these results to larger input dimensions and account for the curse of dimensionality of the Monte Carlo approach to computing the cost function, we believe that the resulting large variability from the MC approach is likely to overshadow any performance difference due to the chosen metamodel.
This further motivates the task of computing the cost function exactly under dependent inputs (see Section~\ref{sec:discussion}).

\subsection{Application to climate simulator}
\label{sec:enroadsclimate}

Here we estimate Shapley effects from data generated from the En-ROADS climate simulator \citep{enroads}. This simulator is a mathematical model of how global temperature is influenced by changes in energy, land use, consumption, agriculture, and other factors. It is designed to be easily used by the general public. The model is an ordinary differential equation solved by Euler integration and synthesizes the important drivers of climate in a computationally efficient and easy-to-use web interface.

\begin{figure}
\centering
\includegraphics[width=\textwidth]{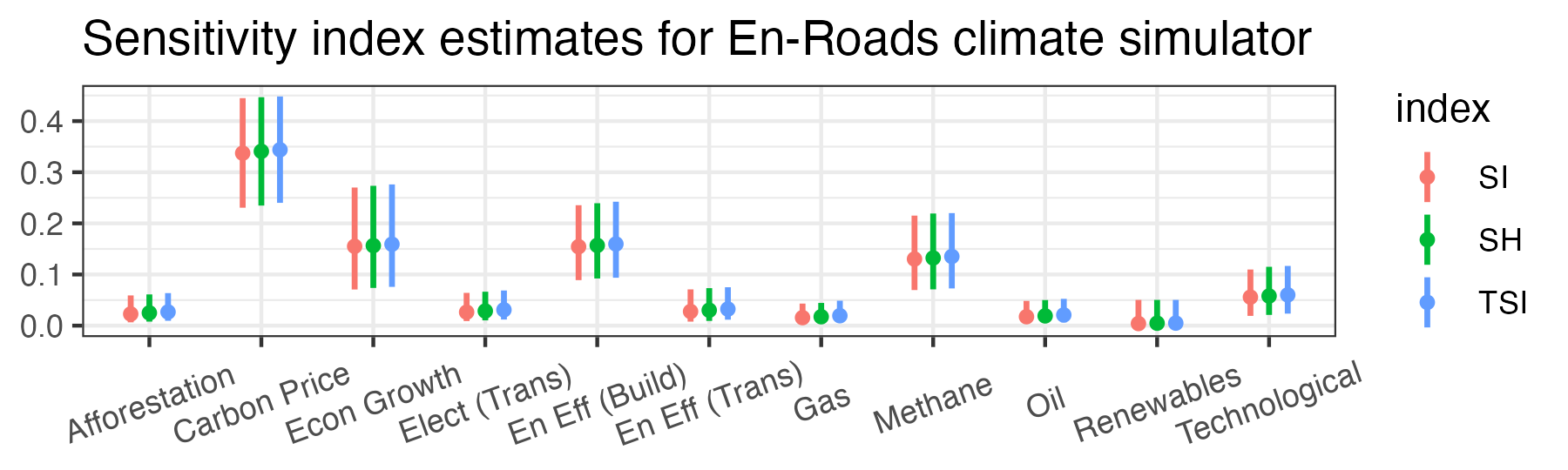}
\caption{$95\%$ credible intervals (as computed in \eqref{eq:shapleydrawrandom}) for (normalized) first-order Sobol\'{} indices (SI), Shapley effects (SH), and total-effect Sobol\'{} indices (TSI) over 1000 posterior draws from a BART model fit to climate simulator data \citep{enroads}.}
\label{fig:enroads}
\end{figure}

The data consists of $n=110$ observations with $p=11$ inputs and was collected using the scheme described in \cite{horiguchi2021assessing}.
To this data we fit a BART model and compute Shapley-effect estimates using the implementation in \cite{openbt} with $1000$ posterior draws, $200$ trees, and the remaining default parameter settings.

Figure~\ref{fig:enroads} shows the estimates for the first-order Sobol\'{} index, Shapley effect, and total-effect Sobol\'{} index of the 11 inputs. 
For each input, the relationship \eqref{eq:sobolshapley} between the three indices is shown. 
As expected given this relationship and the analysis in \cite{horiguchi2021assessing}, the small differences between the Shapley-effect estimates and the two Sobol\'{}-index estimates indicate small interaction effects between any group of inputs. 
Hence, takeaways about the impact of each input are the same as discussed in \cite{horiguchi2021assessing}.
In particular, the four most impactful inputs seem to be carbon price, energy efficiency of buildings, methane, and economic growth.

\section{Discussion}
\label{sec:discussion}

This article establishes posterior contraction rates for Sobol\'{}-index and Shapley-effect estimators computed using BART.
The proofs of our contration rates required proving a property similar to Lipschitz continuity for Sobol\'{} indices and Shapley effects
before using recent contraction-rate results that applies to function spaces with heterogeneous smoothness and sparsity in high dimensions and to fixed and random designs.
This article also illustrates the computational tractability and performance of BART-based Shapley effects on four different test functions under orthogonal inputs and $p=500$.
Code to fit BART models and compute Sobol\'{} index and Shapley effect estimates is found in \cite{openbt}.

Our theoretical consistency results apply to input distributions that are not orthogonal, and thus uncertainty quantification of the Shapley effects would maintain its validity under such distributions.
        However, to implement our approach under such distributions, we would need to be able to compute the cost function. Specifically, the input distribution would affect the values of the probability measure of the boxes that the BART ensemble partitions the input space into. Under independent inputs, the probability measure is simply the volume (i.e., the Lebesgue measure) of the boxes, which is what we currently have implemented. 
        Current Monte Carlo methods of approximating the cost function require being able sample from the input distribution.
        It is more useful (and more challenging) to learn the input distribution based on the observed covariates, and then use this learned distribution to estimate the cost function. 
        For our BART-based approach, this can be possibly achieved by replacing the volume of each hyperrectangle used to compute BART-based Sobol\'{} indices and Shapley effects with the proportion of observations that fall in each hyperrectangle.
        Another possible approach would be to incorporate a tree-based density estimation method suitable for higher-dimensional spaces, such as \cite{awaya2024hidden,awaya2024unsupervised}.
        We will reserve this exploration as future work.

\bibliographystyle{chicago}

\bibliography{mycites}

\newpage
\renewcommand\thesection{\Alph{section}}
\renewcommand\thesubsection{\thesection.\arabic{subsection}}
\setcounter{section}{0}

\bigskip
\begin{center}
{\large\bf SUPPLEMENTARY MATERIAL}
\end{center}

\section{Table summarizing properties of various metamodels under nonparametric regression.}

See Table~\ref{tbl:metamodels}.

\begin{table}
\scriptsize
\centering
\def\arraystretch{0.7}
\begin{tabular}{|c||p{1.6cm}|p{1.7cm}|c|p{1.6cm}|p{2.2cm}|p{1.6cm}|} 
\hline
& Consistency established? & Adapt to discontinuities in regression function? & UQ & Tractable to fit model for $p=250$? & Analytical expression for Shapley effects or Sobol\'{} indices? & Available code to estimate Shapley effects? \\ 
\hline \hline
\parbox[t]{2mm}{\rotatebox[origin=r]{90}{BART }} & yes \citep{jeong2023art} & yes \citep{jeong2023art} & Bayesian & yes (Section~5) & yes \citep{horiguchi2021assessing} & this paper \citep{openbt} \\ \hline
\parbox[t]{2mm}{\rotatebox[origin=r]{90}{ GP }} & yes & yes \citep{mohammadi2019emulating} & Bayesian & no & yes for some covariance kernels & sensitivity R package \citep{iooss2019shapley}\\ \hline
\parbox[t]{2mm}{\rotatebox[origin=r]{90}{ PCE}} & no& no& bootstrap & no & yes \citep{Sudret2008} & no \\ \hline
\parbox[t]{2mm}{\rotatebox[origin=r]{90}{BMARS}} & no& no& Bayesian & yes \citep{francom2018sensitivity} & yes \citep{francom2018sensitivity} & no \\ \hline
\end{tabular}
\caption{Properties of various metamodels under nonparametric regression.}
\label{tbl:metamodels}
\end{table}

\section{Review of posterior contraction}

A posterior contraction rate quantifies how quickly a posterior distribution approaches the true parameter of the data's distribution. 
We use a simplified version of the definition from \cite{ghosal2017fundamentals}:
for every $n \in \mathbb{N}$, let $X^{(n)}$ be an observation in a sample space $(\mathfrak{X}^{(n)}, \mathscr{X}^{(n)})$ with distribution $P_{\theta}^{(n)}$ indexed by $\theta$ belonging to a first countable topological space $\Theta$.
Given a prior $\Pi_n$ on the Borel sets of $\Theta$, let $\Pi_n(\cdot \mid X^{(n)})$ be (a fixed particular version of) the posterior distribution. 
\begin{definition}[Posterior contraction rate]
    A sequence $\{\varepsilon_n\}_{n \in \mathbb{N}}$ is a posterior contraction rate at the parameter $\theta_0$ with respect to the semimetric $d$ if $\Pi_n(\theta\colon d(\theta, \theta_0) \geq M_n \varepsilon_n \mid X^{(n)}) \rightarrow 0$ in $P_{\theta_0}^{(n)}$-probability, for every $M_n \rightarrow \infty$.
\end{definition}
If there exists a constant $M>0$ such that $\Pi_n(\theta\colon d(\theta, \theta_0) \geq M \varepsilon_n \mid X^{(n)}) \rightarrow 0$ in $P_{\theta_0}^{(n)}$-probability, 
then the sequence $\{\varepsilon_n\}_{n \in \mathbb{N}}$ satisfies the definition of posterior contraction rate. 
This will be relevant in interpreting Corollaries \ref{corr:randomdesign} and \ref{corr:fixeddesign} in Section~3.

After reviewing the concept of contraction rates, we state for convenience the conditions made in the theorems of \cite{jeong2023art} that our contraction-rate results rely on.
Because these conditions are not the focus of this paper, we leave discussion of the context behind these conditions to \cite{jeong2023art}.

\section{Preliminaries for proofs}

For any positive integer $m$, denote $[m] \coloneqq \{1, \ldots, m\}$.

From the main text, we copy here the considered regression models with fixed and random design.
The regression model with \textit{fixed} design is
\begin{align}
\label{eq:fixeddesign}
Y_i = f_0(\x_i) + \varepsilon_i, \qquad \varepsilon_i \sim N(0, \sigma^2_0), \qquad i = 1, \ldots, n,
\end{align}
where $\sigma^2_0 < \infty$ and each covariate $\x_i \in [0,1]^p$ is fixed.  
The regression model with \textit{random} design is
\begin{align}
\label{eq:randomdesign}
Y_i = f_0(\X_i) + \varepsilon_i, \qquad \X_i \sim \pi, \qquad \varepsilon_i \sim N(0, \sigma^2_0), \qquad i = 1, \ldots, n,
\end{align}
where $\sigma^2_0 < \infty$, each $\X_i \in [0,1]^p$ is a $p$-dimensional random covariate, and $\pi$ is a probability measure such that $\text{supp}(\pi) \subseteq [0,1]^p$.

\paragraph{Piecewise heterogeneous anisotropic functions}

Next we introduce the conditions of the theorems of \cite{jeong2023art} relevant to our work.
The first set of conditions involves what values of $f_0$ and $\sigma_0^2$ are allowed for BART to contract around $f_0$.
A common assumption for $f_0$ is isotropic smoothness, but this excludes the realistic scenario that $f_0$ is discontinuous and has different degrees of smoothness in different directions and regions. 
\cite{jeong2023art} introduce a new class of \textit{piecewise heterogeneous anisotropic} functions whose domain is partitioned into many boxes (i.e. hyperrectangles), each of which has its own anisotropic smoothness with the same harmonic mean. 
First assume $f_0$ is $d$-sparse, i.e. there exists a function $h_0\colon [0,1]^d \rightarrow \mathbb{R}$ and a subset $S_0 \subseteq [p]$ with $|S_0|=d$ such that $f_0(\x) = h_0(\x_{S_0})$ for any $\x \in [0,1]^p$.
For any given box $\Xi \subseteq [0,1]^d$, smoothness parameter $\bmalpha = (\alpha_1, \ldots, \alpha_d)^T \in (0,1]^d$, and Hölder coefficient $\lambda < \infty$, an \textit{anisotropic $\bmalpha$-Hölder space} on $\Xi$ is defined as 
\begin{equation*}
    \cH_{\lambda}^{\bmalpha,d}(\Xi) \coloneqq \Big\{ h\colon \Xi \rightarrow \mathbb{R}; |h(x)-h(y)| \leq \lambda \sum_{j=1}^d |x_j - y_j|^{\alpha_j}, x,y \in \Xi \Big\}.
\end{equation*}
Though $h_0$ might have different anisotropic smoothness on different boxes, it is important to assume that all boxes have the same harmonic mean. 
Thus define the set $\cA_{\balpha}^{R,d}$ to be the set of $R$-tuples of smoothness parameters that have harmonic mean $\balpha \in (0,1]$:
\begin{equation*}
    \cA_{\balpha}^{R,d} \coloneqq 
    \Big\{ (\bmalpha_1, \ldots, \bmalpha_R): \bmalpha_r \in (0,1]^d, \balpha^{-1} = p^{-1} \sum_{j=1}^d \alpha_{rj}^{-1}, r \in [R] \Big\}.
\end{equation*}
Given a partition $(\Xi_1, \ldots, \Xi_R)$ of $[0,1]^d$ with boxes $\Xi_r \subseteq [0,1]^d$ and a smoothness $R$-tuple $A_{\balpha} \in \cA_{\balpha}^{R,d}$ for some $\balpha \in (0,1]$,
define a \textit{piecewise heterogeneous anisotropic Hölder space} as
\begin{equation*}
    \cH_{\lambda}^{A_{\balpha}, d}(\mathfrak{X})
    \coloneqq \Big \{ h\colon[0,1]^d \rightarrow \mathbb{R}; h|_{\Xi_r} \in \cH_{\lambda}^{\bmalpha_r,d}(\Xi_r), r \in [R] \Big\}.
\end{equation*}
To extend a function from a sparse domain to the original domain $[0,1]^p$, for any nonempty subset $S \subseteq [p]$ define $W_S^p\colon \cC(\bbR^{|S|}) \rightarrow \cC(\bbR^p)$ as the map that extends $h \in \cC(\bbR^{|S|})$ 
to the function $W_S^p h\colon \x \rightarrow h(\x_S)$ where $\x \in [0,1]^p$ and $\cC(E)$ denotes the class of real-valued continuous functions defined on a Euclidean subspace $E$.
With this definition, the space $\cH_{\lambda}^{A_{\balpha}, d}(\mathfrak{X})$ from the preceding panel can be extended to the corresponding \textit{$d$-sparse piecewise heterogeneous anisotropic Hölder space}
\begin{align*}
    \Gamma_{\lambda}^{A_{\balpha}, d, p}(\mathfrak{X}) \coloneqq 
    \bigcup_{S \subseteq [p]\colon \,|S|=d} W_S^p \Big( \cH_{\lambda}^{A_{\balpha}, d}(\mathfrak{X}) \Big).
\end{align*}

With these definitions, we can now state the needed assumptions on the true $f_0$ and $\sigma^2$.
\begin{itemize}
    \item [(A1)] For $d>0$, $\lambda>0$, $R>0$, $\mathfrak{X} = (\Xi_1, \ldots, \Xi_R)$, and $A_{\balpha} \in \cA_{\balpha}^{R,d}$ with $\balpha \in (0,1]$, the true function satisfies $f_0 \in \Gamma_{\lambda}^{A_{\balpha}, d, p}(\mathfrak{X})$ or $f_0 \in \Gamma_{\lambda}^{A_{\balpha}, d, p}(\mathfrak{X}) \cap \mathcal{C}([0,1]^p)$.
    \item [(A2)] It is assumed that $d, p, \lambda, R$, and $\balpha$ satisfy $\epsilon_n \ll 1$, where 
    \begin{align} \label{eq:rate}
        \epsilon_n \coloneqq \sqrt{\frac{d \log p}{n}} + (\lambda d)^{d/(2\balpha+d)} \Big( \frac{R \log n}{n} \Big)^{\balpha/(2\balpha+d)}.
    \end{align}
    \item [(A3)] The true function $f_0$ satisfies $\norm{f_0}_{\infty} \lesssim \sqrt{\log n}$.
    \item [(A4)] The true variance parameter satisfies $\sigma^2 \in [C_0^{-1}, C_0]$ for some sufficiently large $C_0 > 1$.
\end{itemize}

\paragraph{Split-net}

The second set of conditions (of the theorems of \cite{jeong2023art} relevant to our work) involves the split values $c$ allowed in the binary split rules ``$x_j < c$'' of the regression trees.
If a partition of $[0,1]^p$ can be created using the aforementioned tree-based procedure, call it a \textit{flexible tree partition}.
To restrict a flexible tree partition by a set of allowable split values in the binary split rules, for any integer $b_n$ define a \textit{split-net} $\cZ$ to be a finite set of points in $[0,1]^p$ at which possible splits occur along coordinates.
That is, the allowable split values for any input dimension $j \in [p]$ are the $j$th components of the points in the split-net. 
For a given split-net $\cZ$, a flexible tree partition $(\Omega_1, \ldots, \Omega_K)$ of $[0,1]^p$ with boxes $\Omega_k \subseteq [0,1]^p$, $k \in [K]$, is called a \textit{$\cZ$-tree partition} if every split occurs at points in $\cZ$.

A split net should be dense enough for a resulting partition to be close enough to the underlying partition $\mathfrak{X}^* = (\Xi_1^*, \ldots, \Xi_R^*)$ of the true function $f_0$.
For any two box partitions $\mfY^1 = (\Psi^1_1, \ldots, \Psi^1_J)$ and $\mfY^2 = (\Psi^2_1, \ldots, \Psi^2_J)$ with the same number $J$ of boxes, their closeness will be measured using the Hausdorff-type divergence 
\begin{align*}
    \Upsilon(\mfY^1, \mfY^2) \coloneqq \min_{\perm \in \texttt{Perm}[J]} \max_{r \in [J]} \text{Haus}(\Psi_r^1, \Psi^2_{\perm(r)})
\end{align*}
where $\texttt{Perm}[J]$ denotes the set of all permutations of $[J]$ and $\text{Haus}(\cdot,\cdot)$ is the Hausdorff distance. 
For a subset $S \subseteq [p]$, a box partition of $[0,1]^p$ is called \textit{$S$-chopped} if every box $\Psi$ in the box partition satisfies $\max_{j \in S} \texttt{len}([\Psi]_j) < 1$ 
and $\min_{j \notin S} \texttt{len}([\Psi]_j) = 1$, where $[\Psi]_j$ denotes the interval created by projecting the box $[\Psi]$ onto the $j$-th principal axis. 
For a given subset $S \subseteq [p]$, consider an $S$-chopped partition $\mfY$ of $[0,1]^p$ with $J$ boxes. 
For any given $c_n \geq 0$, a split-net $\cZ_n$ is said to be \textit{$(\mfY, c_n)$-dense} if there exists an $S$-chopped $\cZ_n$-tree partition $\T_n$ of $[0,1]^p$ with $J$ boxes such that $\Upsilon(\mfY, \T_n) \leq c_n$.

A split net should also be regular enough (defined below) for a tree partition to capture local features of $f_0$ on each box. 
Assume the underlying partition $\mathfrak{X}^*$ can be approximated well by an $S(\mathfrak{X}^*)$-chopped $\cZ$-tree partition $(\Omega_1^*, \ldots, \Omega_R^*) \coloneqq \arg \min_{\T \in \mathscr{T}_{S(\mathfrak{X}^*),R,\cZ}} \Upsilon(\mathfrak{X}^*,\T)$. 
In each box $\Omega_r^*$, the idea is to allow splits to occur more often along the input dimensions with less smoothness. 
Given a split-net $\cZ$ and splitting coordinate $j$, define the \textit{midpoint-split} of a box $\Psi$ as the bisection of $\Psi$ along coordinate $j$ at the $\lceil \tilde{b}_j(\cZ,\Psi)/2 \rceil$th split-candidate in $[\cZ]_j \cap \texttt{int}([\Psi]_j)$, where $\tilde{b}_j(\cZ,\Psi)$ is the cardinality of $[\cZ]_j \cap \texttt{int}([\Psi]_j)$.
Given a smoothness vector $\bmalpha \in (0,1]^d$, box $\Psi \subseteq [0,1]^p$, split-net $\cZ$, integer $L>0$, and index set $S = \{s_1, \ldots, s_d\} \subseteq [p]$, 
define the \textit{anisotropic k-d tree} $AKD(\Psi; \cZ,\bmalpha,L,S)$ as the iterative splitting procedure that partitions $\Psi$ into disjoint boxes $\Omega_1^{\circ}, \ldots, \Omega_{2^{L^{\circ}}}^{\circ}$ as follows:
\begin{enumerate}
    \item Set $\Omega_1^{\circ} = \Psi$ and set counter $l_j = 0$ for each $j \in [d]$.
    \item Let $L^{\circ} = \sum_{j=1}^d l_j$ for the current counters. For splits at iteration $1 + L^{\circ}$, choose $j' = \min \{ \arg \min_j l_j \alpha_j\}$. Midpoint-split all boxes $\Omega_1^{\circ}, \ldots, \Omega_{2^{L^{\circ}}}^{\circ}$ with the given $\cZ$ and splitting coordinate $s_{j'}$. Relabel the generated new boxes as $\Omega_1^{\circ}, \ldots, \Omega_{2^{1+L^{\circ}}}^{\circ}$, and then increment $l_{j'}$ by one.
    \item Repeat step 2 until either the updated $L^{\circ}$ equals $L$ or the midpoint-split is no longer available. Return counters $l_1,\ldots,l_d$ and boxes $\Omega_1^{\circ}, \ldots, \Omega_{2^{L^{\circ}}}^{\circ}$.
\end{enumerate}
For a given box $\Psi \subseteq [0,1]^p$, smoothness vector $\bmalpha \in (0,1]^d$, integer $L>0$, and index set $S = \{s_1, \ldots, s_d\} \subseteq [p]$, a split-net $\cZ$ is called \textit{$(\Psi,\bmalpha,L,S)$-regular} if the counters and boxes returned by $AKD(\Psi; \cZ,\bmalpha,L,S)$ satisfy $L^{\circ}=L$ and $\max_k \texttt{len}([\Omega_k^{\circ}]_{s_j}) \lesssim \texttt{len} ([\Psi]_{s_j}) 2^{-l_j}$ for every $j \in [d]$.

With these definitions, we can now state the needed assumptions on the sequence $\{\cZ_n\}_{n=1}^{\infty}$ of split-nets.
\begin{itemize}
    \item [(A5)] Each split-net $\cZ_n$ satisfies $\max_{1\leq j\leq p} \log b_j(\cZ_n) \lesssim \log n$, where $b_j(\cZ_n)$ is the cardinality of the set $\{ z_j: (z_1, \ldots, z_p) \in \cZ_n\}$.
    \item [(A6)] Each split-net $\cZ_n$ is suitably dense and regular to construct a $\cZ_n$-tree partition $\hat{\T}$ such that there exists a simple function $\hat{f}_0 \in \mathcal{F}_{\hat{\T}}$ satisfying $\|f_0 - \hat{f}_0 \|_n \lesssim \bar{\epsilon}_n$, where 
    \begin{equation}  \label{eq:barepsilon}
        \bar{\epsilon}_n \coloneqq (\lambda d)^{d/(2\balpha+d)} \{(R \log n)/n\}^{\balpha/(2\balpha+d)},
    \end{equation}
    the empirical $L_2$-norm $\|\cdot\|_n$ is defined as $\|f\|_n^2 = n^{-1} \sum_{i=1}^n |f(\x_i)|^2$, and $\mathcal{F}_{\hat{\T}}$ is the set of functions on $[0,1]^p$ that are constant on each piece of the partition $\hat{\T}$. 
    \item [(A7)] Each $\cZ_n$-tree partition $(\Omega_1^*, \ldots, \Omega_R^*)$ approximating the underlying partition $\mathfrak{X}^*$ for the true function $f_0$ satisfies $\max_{r \in [R]} \texttt{depth}(\Omega_r^*) \lesssim \log n$, where $\texttt{depth}$ means the depth of a node (i.e. number of nodes in the path from that node to the root node).
\end{itemize}

Finally, we state the required prior specification.
\begin{itemize}
    \item [(P1)] Each tree partition in the ensemble is independently assigned a tree prior with Dirichlet sparsity from \cite{Linero18}. This sparse Dirichlet prior places a Dirichlet prior on the proportion vector used to select the splitting coordinate $j$ during the creation of a split rule. 
    \item [(P2)] The step-heights of the regression-tree functions are each assigned a normal prior with mean zero and covariance matrix whose eigenvalues are bounded below and above.
    \item [(P3)] The variance parameter $\sigma^2$ is assigned an inverse gamma prior.
\end{itemize}

\cite{jeong2023art} make the above assumptions and prior specification for their contraction-rate results in the fixed design setting \eqref{eq:fixeddesign}.
For their contraction-rate results in the random design setting \eqref{eq:randomdesign}, a few of the above assumptions and prior specifications are replaced by the following:
\begin{itemize}
    \item [(A3$\ast$)] The true function $f_0$ satisfies $\norm{f_0}_{\infty} \leq C_0^*$ for some sufficiently large $C_0^* > 0$.
    \item [(A6$\ast$)] The split-net $\cZ$ is suitably dense and regular to construct a $\cZ$-tree partition $\hat{\T}$ such that there exists $\hat{f}_0 \in \mathcal{F}_{\hat{\T}}$ satisfying $\|f_0 - \hat{f}_0\| \lesssim \bar{\epsilon}_n$ where $\bar{\epsilon}_n$ is given by \eqref{eq:barepsilon}.
    \item [(P2$\ast$)] A prior on the compact support $[-\bar{C}_1, \bar{C}_1]$ is assigned to the step-heights of the regression-tree functions for some $\bar{C}_1 > C_0^*$.
    \item [(P3$\ast$)] A prior on the compact support $[\bar{C}_2^{-1}, \bar{C}_2]$ is assigned to the variance parameter $\sigma^2$ for some $\bar{C}_2 > C_0$.
\end{itemize}

\section{Posterior asymptotics}
\label{sec:mainresults}

This section establishes our contraction-rate results (Corollaries \ref{corr:randomdesign} and \ref{corr:fixeddesign}) for estimators of Sobol\'{} indices and Shapley effects under either the fixed design \eqref{eq:fixeddesign} or the random design \eqref{eq:randomdesign}.
Our proofs rely on these sensitivity indices having a property (defined in Lemma~\ref{lemma:two} below) similar to but slightly less restrictive than Lipschitz continuity.
However, the tasks of proving this property for all of these sensitivity indices are very similar to each other.
Because these indices are linear combinations of the functional $c_{P, \pi}$ defined in \eqref{eq:c}, we can use Lemma~\ref{lemma:two} to reduce the above tasks to the single task of proving this property for $c_{P, \pi}$.

\begin{lemma}
\label{lemma:two}
Suppose the following relationship is true for all indices $k$ in a finite set $\cA$:
given two metric spaces $X$ and $X_0$ with the same metric $d_X$,
there exists a constant $C>0$ such that, for all $(x, x_0) \in X \times X_0$,
the function $\phi_k\colon X \cup X_0 \rightarrow \R$ satisfies
\begin{align*}
|\phi_{k}(x) - \phi_{k}(x_0)| \leq C d_X(x, x_0).
\end{align*}
Then any set $\{a_k\}_{k \in \cA}$ of real numbers satisfies
\begin{align*}
\Bigl\lvert \sum_{k \in \cA} a_k \phi_k(x) - \sum_{k \in \cA} a_k \phi_k(x_0) \Bigr\rvert
&\leq C^* d_X(x, x_0).
\end{align*}
where $C^* = C \sum_{k \in \cA_+} |a_k|$ and $\cA_+ \coloneqq \{k \in \cA \colon \bigl\lvert \phi_k(x) - \phi_k(x_0) \bigr\rvert > 0\}$. 
\end{lemma}

\subsection{Nonparametric regression with random design}

This section assumes the random-design regression setting \eqref{eq:randomdesign}; 
all expectations in this section are with respect to the probability measure $\pi$ in \eqref{eq:randomdesign}.

\begin{theorem}
\label{thm:randomdesign}
Assume (A3*).
If $f \in L^2([0,1]^p)$ shares the same bound $C_0^*$ from (A3*), then for any subset $P \subseteq [p]$ and distribution $\pi$ with support $[0,1]^p$ we have
\begin{align*}
|c_{P, \pi}(f) - c_{P, \pi}(f_0)| 
&\leq 4C_0^* \norm{f - f_0}_{2,\pi}
\end{align*}
for the functional $c_{P, \pi}$ defined in \eqref{eq:c}.
\end{theorem}

\begin{corollary}
\label{corr:randomdesign}
Under the assumptions of Theorem 4 of \cite{jeong2023art} -- Assumptions (A1), (A2), (A3$\ast$), (A4), (A5), (A6$\ast$), and (A7), and the prior assigned through (P1), (P2$\ast$), and (P3$\ast$) -- and Theorem~\ref{thm:randomdesign} above, there exist positive constants $L_{V,\pi,|P|}$, $L_{T,\pi}$, and $L_S$ such that 
as $n \rightarrow \infty$ for $\epsilon_n$ in \eqref{eq:rate}, 
\begin{align*}
\E_0 \Pi \Big\{(f, \sigma^2)\colon |V_{P, \pi}(f) - V_{P, \pi}(f_0)| + |\sigma^2 - \sigma^2_0| > L_{V,\pi,|P|} \epsilon_n  \Big| Y_1, \ldots, Y_n \Big\} &\rightarrow 0, \\
\E_0 \Pi \Big\{(f, \sigma^2)\colon |T_{j, \pi}(f) - T_{j, \pi}(f_0)| + |\sigma^2 - \sigma^2_0| > L_{T,\pi} \epsilon_n  \Big| Y_1, \ldots, Y_n \Big\} &\rightarrow 0, \\ 
\text{ and } \E_0 \Pi \Big\{(f, \sigma^2)\colon |S_{j, \pi}(f) - S_{j, \pi}(f_0)| + |\sigma^2 - \sigma^2_0| > L_S \epsilon_n  \Big| Y_1, \ldots, Y_n \Big\} &\rightarrow 0.
\end{align*}
\end{corollary}

\subsection{Nonparametric regression with fixed design}

This section assumes the fixed-design regression setting \eqref{eq:fixeddesign}; 
all expectations in this section are with respect to the probability measure $P_{\sX}(\cdot) = n^{-1} \sum_{\x \in \sX} \delta_{\x} (\cdot)$ where 
$\sX$ is the set of the fixed covariates assumed in \eqref{eq:fixeddesign}.

\begin{theorem}
\label{thm:fixeddesign}
Assume (A3).
If $f \in L^2([0,1]^p)$ shares the same bound $\sqrt{\log n}$ from (A3), then for any subset $P \subseteq [p]$ and distribution $\pi$ with support $[0,1]^p$ we have
\begin{align*}
|c_{P, P_{\sX}}(f) - c_{P, P_{\sX}}(f_0)| 
&\lesssim 4 \sqrt{\log n} \norm{f - f_0}_{2,P_{\sX}}
\end{align*}
where the empirical $L_2$-norm $\norm{\cdot}_{2,P_{\sX}}$ is defined as $\norm{f}_{2,P_{\sX}}^2 = n^{-1} \sum_{\x \in \sX} |f(\x)|^2$.
\end{theorem}

\begin{corollary}
\label{corr:fixeddesign}
Under the assumptions of Theorem 2 of \cite{jeong2023art} -- Assumptions (A1), (A2), (A3), (A4), (A5), (A6), and (A7), and the prior assigned through (P1), (P2), and (P3) -- and Theorem \ref{thm:fixeddesign} above, there exist positive constants $L_{V,\pi,|P|}$, $L_{T,\pi}$, and $L_S$ such that
as $n \rightarrow \infty$ for $\epsilon_n$ in \eqref{eq:rate},
\begin{align*}
\E_0 \Pi \Big\{(f, \sigma^2)\colon |V_{P, \pi}(f) - V_{P, \pi}(f_0)| + |\sigma^2 - \sigma^2_0| > L_{V,\pi,|P|} \epsilon_n \sqrt{\log n} \mid Y_1, \ldots, Y_n \Big\} &\rightarrow 0, \\
\E_0 \Pi \Big\{(f, \sigma^2)\colon |T_{j, \pi}(f) - T_{j, \pi}(f_0)| + |\sigma^2 - \sigma^2_0| > L_{T,\pi} \epsilon_n \sqrt{\log n} \mid Y_1, \ldots, Y_n \Big\} &\rightarrow 0, \\ 
\text{ and } \E_0 \Pi \Big\{(f, \sigma^2)\colon |S_{j, \pi}(f) - S_{j, \pi}(f_0)| + |\sigma^2 - \sigma^2_0| > L_S \epsilon_n \sqrt{\log n} \mid Y_1, \ldots, Y_n \Big\} &\rightarrow 0.
\end{align*}
\end{corollary}

\section{Proofs of results in main text}

For convenience, here we replicate the theorems, lemmas, and relevant equations from the main text.

\subsection{Relevant quantities}

\begin{equation}
\label{eq:c}
c_{P,\pi}(f) = \V_{\pi} [\E_{\pi} \{f(\X) \mid \X_P\}] = \E_{\pi} ([\E_{\pi} \{f(\X) \mid \X_P\}]^2) - [\E_{\pi} \{f(\X)\}]^2.
\end{equation}

\begin{align} \label{eq:rate}
    \epsilon_n \coloneqq \sqrt{\frac{d \log p}{n}} + (\lambda d)^{d/(2\balpha+d)} \Big( \frac{R \log n}{n} \Big)^{\balpha/(2\balpha+d)}.
\end{align}

\begin{equation} \label{eq:shapley}
    S_{j, \pi}(f) = (p!)^{-1} \sum_{P \subseteq ([p] \setminus \{j\})} (p - |P| - 1)! \, |P|! \, \big\{c_{P \cup \{j\},\pi}(f) - c_{P,\pi}(f)\big\}.
\end{equation}

\subsection{Theorems and lemmas}

\begin{lemma}
\label{lemma:two}
Suppose the following relationship is true for all indices $k$ in a finite set $\cA$:
given two metric spaces $X$ and $X_0$ with the same metric $d_X$,
there exists a constant $C>0$ such that, for all $(x, x_0) \in X \times X_0$,
the function $\phi_k\colon X \cup X_0 \rightarrow \R$ satisfies
\begin{align*}
|\phi_{k}(x) - \phi_{k}(x_0)| \leq C d_X(x, x_0).
\end{align*}
Then any set $\{a_k\}_{k \in \cA}$ of real numbers satisfies
\begin{align*}
\Bigl\lvert \sum_{k \in \cA} a_k \phi_k(x) - \sum_{k \in \cA} a_k \phi_k(x_0) \Bigr\rvert
&\leq C^* d_X(x, x_0).
\end{align*}
where $C^* = C \sum_{k \in \cA_+} |a_k|$ and $\cA_+ \coloneqq \{k \in \cA \colon \bigl\lvert \phi_k(x) - \phi_k(x_0) \bigr\rvert > 0\}$. 
\end{lemma}

\begin{theorem}
\label{thm:randomdesign}
Assume (A3*).
If $f \in L^2([0,1]^p)$ shares the same bound $C_0^*$ from (A3*), then for any subset $P \subseteq [p]$ and distribution $\pi$ with support $[0,1]^p$ we have
\begin{align*}
|c_{P, \pi}(f) - c_{P, \pi}(f_0)| 
&\leq 4C_0^* \norm{f - f_0}_{2,\pi}
\end{align*}
for the functional $c_{P, \pi}$ defined in \eqref{eq:c}.
\end{theorem}

\begin{corollary}
\label{corr:randomdesign}
Under the assumptions of Theorem 4 of \cite{jeong2023art} -- Assumptions (A1), (A2), (A3$\ast$), (A4), (A5), (A6$\ast$), and (A7), and the prior assigned through (P1), (P2$\ast$), and (P3$\ast$) -- and Theorem~\ref{thm:randomdesign} above, there exist positive constants $L_{V,\pi,|P|}$, $L_{T,\pi}$, and $L_S$ such that 
as $n \rightarrow \infty$ for $\epsilon_n$ in \eqref{eq:rate}, 
\begin{align*}
\E_0 \Pi \Big\{(f, \sigma^2)\colon |V_{P, \pi}(f) - V_{P, \pi}(f_0)| + |\sigma^2 - \sigma^2_0| > L_{V,\pi,|P|} \epsilon_n  \Big| Y_1, \ldots, Y_n \Big\} &, \\
\E_0 \Pi \Big\{(f, \sigma^2)\colon |T_{j, \pi}(f) - T_{j, \pi}(f_0)| + |\sigma^2 - \sigma^2_0| > L_{T,\pi} \epsilon_n  \Big| Y_1, \ldots, Y_n \Big\} &, \\ 
\text{ and } \E_0 \Pi \Big\{(f, \sigma^2)\colon |S_{j, \pi}(f) - S_{j, \pi}(f_0)| + |\sigma^2 - \sigma^2_0| > L_S \epsilon_n  \Big| Y_1, \ldots, Y_n \Big\} &
\end{align*}
each shrink to zero.
\end{corollary}

\begin{theorem}
\label{thm:fixeddesign}
Assume (A3).
If $f \in L^2([0,1]^p)$ shares the same bound $\sqrt{\log n}$ from (A3), then for any subset $P \subseteq [p]$ and distribution $\pi$ with support $[0,1]^p$ we have
\begin{align*}
|c_{P, P_{\sX}}(f) - c_{P, P_{\sX}}(f_0)| 
&\lesssim 4 \sqrt{\log n} \norm{f - f_0}_{2,P_{\sX}}
\end{align*}
where the empirical $L_2$-norm $\norm{\cdot}_{2,P_{\sX}}$ is defined as $\norm{f}_{2,P_{\sX}}^2 = n^{-1} \sum_{\x \in \sX} |f(\x)|^2$.
\end{theorem}

\begin{corollary}
\label{corr:fixeddesign}
Under the assumptions of Theorem 2 of \cite{jeong2023art} -- Assumptions (A1), (A2), (A3), (A4), (A5), (A6), and (A7), and the prior assigned through (P1), (P2), and (P3) -- and Theorem \ref{thm:fixeddesign} above, there exist positive constants $L_{V,\pi,|P|}$, $L_{T,\pi}$, and $L_S$ such that
as $n \rightarrow \infty$ for $\epsilon_n$ in \eqref{eq:rate},
\begin{align*}
\E_0 \Pi \Big\{(f, \sigma^2)\colon |V_{P, \pi}(f) - V_{P, \pi}(f_0)| + |\sigma^2 - \sigma^2_0| > L_{V,\pi,|P|} \epsilon_n \sqrt{\log n} \mid Y_1, \ldots, Y_n \Big\} &, \\
\E_0 \Pi \Big\{(f, \sigma^2)\colon |T_{j, \pi}(f) - T_{j, \pi}(f_0)| + |\sigma^2 - \sigma^2_0| > L_{T,\pi} \epsilon_n \sqrt{\log n} \mid Y_1, \ldots, Y_n \Big\} &, \\ 
\text{ and } \E_0 \Pi \Big\{(f, \sigma^2)\colon |S_{j, \pi}(f) - S_{j, \pi}(f_0)| + |\sigma^2 - \sigma^2_0| > L_S \epsilon_n \sqrt{\log n} \mid Y_1, \ldots, Y_n \Big\} &
\end{align*}
each shrink to zero.
\end{corollary}

\subsection{Proofs}

\begin{proof}[Proof of Lemma~\ref{lemma:two}]
We have
\begin{equation*}
\Bigl\lvert \sum_{k \in \cA} a_k \phi_k(x) - \sum_{k \in \cA} a_k \phi_k(x_0) \Bigr\rvert \leq \sum_{k \in \cA} |a_k| \, \bigl\lvert \phi_k(x) - \phi_k(x_0) \bigr\rvert 
\leq \sum_{k \in \cA_+} |a_k| \, C d_X(x, x_0)
\end{equation*}
where the right-most sum in the preceding panel is exactly $C^* d_X(x, x_0)$.
\end{proof}

\begin{proof}[Proof of Theorem~\ref{thm:randomdesign}]
Note that 
\begin{align*}
\Bigl\lvert c_{P, \pi}(f) - c_{P, \pi}(f_0) \Bigr\rvert
&= \Bigl\lvert\Big( \E \big[ (\E [f(\X) \mid \X_P])^2 - (\E [f_0(\X) \mid \X_P])^2 \big] \Big) \\
&\qquad- \Big( [\E f(\X)]^2 - [\E f_0(\X)]^2 \Big) \Bigr\rvert \\
&\leq \E \Bigl\lvert (\E [f(\X) \mid \X_P])^2 - (\E [f_0(\X) \mid \X_P])^2 \Bigr\rvert \\
&\qquad+ \Bigl\lvert  [\E \{f(\X)\}]^2 - [\E \{f_0(\X)\}]^2  \Bigr\rvert. 
\end{align*}
From the assumption that $f$ and $f_0$ are bounded in supremum norm by $C_0^*$, we get
\begin{align*}
&\Bigl\lvert[\E \{f(\X)\}]^2 - [\E \{f_0(\X)\}]^2\Bigr\rvert \\
&\quad= \Bigl\lvert \big[\E \{f(\X)\} + \E\{ f_0(\X)\}\big]\big[\E\{ f(\X)\} - \E\{ f_0(\X)\}\big] \Bigr\rvert \\
&\quad\leq 2 C_0^* \E |f(\X) - f_0(\X)|.
\end{align*}
We can similarly deduce for any $\X_P$ that 
\begin{align*}
\Bigl\lvert (\E [f(\X) \mid \X_P])^2 - (\E [f_0(\X) \mid \X_P])^2 \Bigr\rvert
&\leq 2 C_0^* \; \E [|f(\X) - f_0(\X)| \mid \X_P].
\end{align*}
Then
\begin{align*}
&|c_{P, \pi}(f) - c_{P, \pi}(f_0)| \\
&\leq \E \big[ 2C_0^* \E[|f(\X) - f_0(\X)| \mid \X_P] \big] + 2C_0^* \E |f(\X) - f_0(\X)| \\
&= 4 C_0^* \E |f(\X) - f_0(\X)|.
\end{align*}
To finish, Jensen's inequality implies $ \E |f(\X) - f_0(\X)| \leq \norm{f-f_0}_{2,\pi}$.
\end{proof}

\begin{proof}[Proof of Corollary~\ref{corr:randomdesign}]
Below is the proof just for the $j$th (where $j \in [p]$) total-effect Sobol\'{} index.
The same argument can be followed to obtain the corresponding results for any main-effect Sobol\'{} index and any Shapley effect after making the appropriate substitutions for the $a_P$ below.
Lemma~\ref{lemma:two} and Theorem~\ref{thm:randomdesign} together imply 
\begin{align*}
|T_{j, \pi}(f) - T_{j, \pi}(f_0)| 
&\leq D_{T,\pi} \norm{f - f_0}_{2,\pi}.
\end{align*}
where $D_{T,\pi} \leq \max\{1, 4C_0^* \sum_{P \in [p]} |a_{P,\pi}|\}$ and the real values $a_P$ are the coefficients corresponding to $T_{j, \pi}$ expressed as a linear combination of $c_{P, \pi}$.
(Theorem~\ref{thm:sum} provides upper bounds for the sum $\sum_{P \in [p]} |a_{P,\pi}|$.)
For any constant $\delta > 0$, define the two sets
\begin{align*}
    A_{\delta} &\coloneqq \{ (f, \sigma^2) \colon \lvert T_{j, \pi}(f) - T_{j, \pi}(f_0) \rvert + |\sigma^2 - \sigma^2_0| > \delta \} \\
    B_{\delta} &\coloneqq \{ (f, \sigma^2) \colon D_{T,\pi} \lvert T_{j, \pi}(f) - T_{j, \pi}(f_0) \rvert + D_{T,\pi} |\sigma^2 - \sigma^2_0| > \delta \}.
\end{align*}
Because $D_{T,\pi} \geq 1$, we have $A_{\delta} \subseteq B_{\delta}$ for all $\delta>0$.
Let $\mathscr{D}_n \coloneqq \{(X_1, Y_1), \ldots, (X_n, Y_n)\}$.
By Theorem~4 of \cite{jeong2023art}, there exists a constant $M>0$ such that $\mathbb{E}_0 \Pi (B_{L_{T,\pi} \epsilon_n} \mid \mathscr{D}_n) \rightarrow 0$ as $n \rightarrow \infty$, where $L_{T,\pi} = D_{T,\pi} M$.
Because $A_{L_{T,\pi} \epsilon_n} \subseteq B_{L_{T,\pi} \epsilon_n}$ for all $n$, we have $\mathbb{E}_0 \Pi (A_{L_{T,\pi} \epsilon_n} \mid \mathscr{D}_n) \rightarrow 0$ as $n \rightarrow \infty$.
\end{proof}

The proofs of Theorem \ref{thm:fixeddesign} and Corollary \ref{corr:fixeddesign} can be obtained by replacing the random-design bound $C_0^*$ with $\sqrt{\log n}$ and the distribution $\pi$ with the probability measure $P_{\sX}$.

Regarding the constant $D_{T,\pi}$ (and the corresponding constants $D_{V,\pi,|P|}$ and $D_{T,\pi}$) in the proof of Corollary~\ref{corr:randomdesign}, the sum $\sum_{P \in [p]} |a_{P,\pi}|$ seems to grow exponentially in $p$. Theorem~\ref{thm:sum} below states that this exponential dependence on $p$ holds really only for a total-effect Sobol\'{} index (although the sum for a Sobol\'{} index $V_P$ is $2^{|P|}-1$, in practice such indices are computed only for $|P| \leq 3$). 
However, $p$ is often much larger than the order of the highest-order interaction in the true function. 
If the input distribution $\pi$ is orthogonal (which is needed for a Sobol\'{} index to be interpretable), if the true function does not contain interactions of order larger than $q \leq p$, and if the BART posterior assigns zero probability to functions with interactions of order larger than $q$ (this third assumption is not unreasonable for even moderately large $q$, given that BART's prior discourages deep trees and a tree's regression function cannot have interactions of order larger than the tree's depth), then the sum's dependence on $p$ for the total-effect index reduces to an exponential dependence on $q$, which is often quite small.
(We can further reduce this dependence on $q$ if $\pi$ is orthogonal by omitting Sobol\'{} index terms for subsets containing inert variables in a similar fashion as described in the proof of Theorem~\ref{thm:sum}.)

\begin{theorem} \label{thm:sum}
    Upper bounds for 
    $D_S$, $D_{V,\pi,|P|}$, and $D_{T,\pi}$ in the proof of Corollary~\ref{corr:randomdesign}
    are, respectively, $\max\{1,8C_0^*\}$, $\max\{1,4C_0^*(2^{|P|}-1)\}$, and $\max\{1,4C_0^*\sum_{i=0}^{p-1} (p-1)!/\{(p-1-i)!i!\} (2^{i+1}-1)\}$.
    If $\pi$ is orthogonal and neither $f$ nor $f_0$ in Corollary~\ref{corr:randomdesign} contain interactions of order larger than $q<p$, then the preceding upper bounds for $D_{V,\pi,|P|}$ and $D_{T,\pi}$ can be reduced to, respectively, 
    $\max\{1,4C_0^*\sum_{i=1}^{\min\{q,|P|\}} P!/\{(P-i)!i!\}\}$ and $\max\{1,4C_0^*\sum_{i=0}^{q-1} (q-1)!/\{(q-1-i)!i!\} (2^{i+1}-1)\}$. 
\end{theorem}

\begin{proof}[Proof of Theorem~\ref{thm:sum}]
For the Shapley effect bound, we note that
\begin{align*}
    \frac{1}{p}\sum_{P \subseteq ([p] \setminus \{j\})} \frac{(p-1-|P|)!|P|!}{(p-1)!} = \frac{1}{p}\sum_{i=0}^{p-1} 1 = 1.
\end{align*}
This with \eqref{eq:shapley} and Theorem~\ref{thm:randomdesign} together imply
\begin{align*}
    \Big|S_{j, \pi}(f) - S_{j, \pi}(f_0)\Big| 
    &\leq \left|\frac{1}{p} \sum_{P \subseteq ([p] \setminus \{j\})} \frac{(p-1-|P|)!|P|!}{(p-1)!} \left[c_{P \cup \{j\},\pi}(f) - c_{P \cup \{j\},\pi}(f_0) \right]\right| \\
    &+ \left|\frac{1}{p} \sum_{P \subseteq ([p] \setminus \{j\})} \frac{(p-1-|P|)!|P|!}{(p-1)!} \left[c_{P,\pi}(f) - c_{P,\pi}(f_0) \right]\right| \\
    &\leq 4C_0^* \norm{f - f_0}_{2,\pi} + 4C_0^* \norm{f - f_0}_{2,\pi}.
\end{align*}

As defined in Section~2.2, a Sobol\'{} index $V_{P,\pi}(f)$ is a linear combination of costs \eqref{eq:c} over all nonempty subsets of $P$, where each coefficient in the linear combination is either $1$ or $-1$.
Since $P$ has $\sum_{i=1}^{|P|} |P|!/\{(|P|-i)!i!\} = 2^{|P|}-1$ many nonempty subsets, we can use Theorem~\ref{thm:randomdesign}  to get
\begin{align*}
    \Big|V_{P, \pi}(f) - V_{P, \pi}(f_0)\Big| 
    &\leq (2^{|P|}-1) 4C_0^* \norm{f - f_0}_{2,\pi}.
\end{align*}
If $\pi$ is orthogonal and there are no interactions of order larger than $q$, then the Sobol\'{} indices for the subsets of $P$ containing more than $q$ elements are zero, and hence we can omit those Sobol\'{} indices from $V_{P,\pi}(f) - V_{P,\pi}(f_0)$.
Since $P$ has $\sum_{i=1}^{\min\{q,|P|\}} |P|!/\{(|P|-i)!i!\}$ many nonempty subsets containing at most $q$ elements, the desired result follows.

As defined in Section~2.2, a total-effects Sobol\'{} index $T_{j,\pi}(f)$ is the sum of $V_{P,\pi}(f)$ over all subsets $P \subseteq [p]$ containing $j$. Using the above result, we get
\begin{align*}
    \Big|T_{j, \pi}(f) - T_{j, \pi}(f_0)\Big| 
    &= \left|\sum_{P \subseteq ([p] \setminus \{j\})} V_{P\cup\{j\},\pi}(f) - V_{P\cup\{j\},\pi}(f_0) \right| \\
    &\leq \sum_{P \subseteq ([p] \setminus \{j\})} (2^{|P|+1}-1) 4C_0^* \norm{f - f_0}_{2,\pi}\\
    &\leq \sum_{i=0}^{p-1} \frac{(p-1)!}{(p-1-i)!i!} (2^{i+1}-1) 4C_0^* \norm{f - f_0}_{2,\pi}.
\end{align*}
If $\pi$ is orthogonal and there are no interactions of order larger than $q$, then  the remaining result follows if we again omit from each sum the subsets of $P$ containing more than $q$ elements.
\end{proof}

\section{Functions used in simulation studies in Section~5}

Table \ref{tbl:VA1} contains the variances, Sobol\'{} indices, and Shapley values for each test function.
\begin{enumerate}
    \item The ``Friedman'' function \citep{Friedman91} is defined as \[f(\x) \coloneqq 10\sin(\pi x_1 x_2) + 20(x_3 - 0.5)^2 + 10x_4 + 5x_5.\]
    \item The ``Morris'' function inspired by \cite{Morris2006} is defined as 
        \begin{align*}
        f(\x) &\coloneqq \alpha \sum_{i=1}^{d} x_i + \beta \sum_{i=1}^{d-1} x_i \sum_{j=i+1}^{d}  x_j 
        \end{align*}
        where $\alpha = \sqrt{12} - 6\sqrt{0.1 (d-1)} \approx -0.331$ and $\beta = \frac{12}{\sqrt{10(d-1)}} \approx 1.897$ are chosen. 
    \item The ``Bratley'' function \citep{Bratley1992,Kucherenko2011} is defined as 
    \begin{align*}
        f(\x) &\coloneqq \sum_{i=1}^{d} (-1)^i \prod_{j=1}^i x_j 
        = -x_1 + x_1 x_2 - x_1 x_2 x_3 + x_1 x_2 x_3 x_4 - x_1 x_2 x_3 x_4 x_5.
    \end{align*}
    \item The ``$g$-function'' from \cite{Saltelli95} is defined as \[f(\x) \coloneqq \prod_{k=1}^{d} \frac{|4x_k-2| + c_k}{1+c_k},\] where we use $c_k = (k-1)/2$ for $k = 1, \ldots, d$ suggested by \cite{Crestaux09}. 
\end{enumerate}

\begin{table}
\small
\centering
\def\arraystretch{0.6}
\begin{tabular}{|r || *9{c|} *3{c|}} 
\hline 
& \multicolumn{3}{c|}{Friedman} & \multicolumn{3}{c|}{Morris} & \multicolumn{3}{c|}{Bratley} & \multicolumn{3}{c|}{$g-$function} \\ \hline  
  & \multicolumn{3}{c|}{Var: $23.8$} & \multicolumn{3}{c|}{Var: $5.25$} & \multicolumn{3}{c|}{Var: $0.057$} & \multicolumn{3}{c|}{Var: $3.076$} \\ \hline
$j$ & $V_j^*$ & $T_j^*$ & $S_j^*$ & $V_j^*$ & $T_j^*$ & $S_j^*$ & $V_j^*$ & $T_j^*$ & $S_j^*$ & $V_j^*$ & $T_j^*$ & $S_j^*$\\ 
\hline \hline
1  & 0.197 & 0.274 & 0.235 & 0.190 & 0.210 & 0.2 & 0.688 & 0.766 & 0.725 & 0.411 & 0.558 & 0.482 \\
2  & 0.197 & 0.274 & 0.235 & 0.190 & 0.210 & 0.2 & 0.142 & 0.220 & 0.179 & 0.183 & 0.288 & 0.233 \\
3  & 0.093 & 0.093 & 0.093 & 0.190 & 0.210 & 0.2 & 0.051 & 0.099 & 0.073 & 0.103 & 0.172 & 0.135 \\
4  & 0.350 & 0.350 & 0.350 & 0.190 & 0.210 & 0.2 & 0.006 & 0.018 & 0.011 & 0.066 & 0.113 & 0.088 \\
5  & 0.087 & 0.087 & 0.087 & 0.190 & 0.210 & 0.2 & 0.006 & 0.018 & 0.011 & 0.046 & 0.080 & 0.062 \\
\hline
\end{tabular}
\caption{Normalized main-effects $V_j^*=V_j^*(f)$, total-effects $T_j^*=T_j^*(f)$, and Shapley effects $S_j^*=S_j^*(f)$ for various functions $f$ and variable indices $j \in [5]$ under orthogonal inputs.}
\label{tbl:VA1}
\end{table}

\section{How do metamodels scale with input dimension?}
\label{sec:scale-with-input-dimension}

Here we explore how various metamodels scale with input dimension $p$.
Sparse variational GPs are known for being scalable in sample size $n$, but as explained in the abstract of \cite{burt2020convergence}, to make the KL-divergence between the approximate model and the exact posterior arbitrarily small for a Gaussian-noise regression model with $M << N$ inducing points, squared-exponential kernel, and $p$-dimensional Gaussian distributed covariates, an overall computational cost of $O(N(\log N )^{2p} (\log \log N )^2)$, which is exponential in the input dimension $p$, is required. 

Deep GPs \citep{damianou2013deep,sauer2023active} are meant to be scalable in sample size, but not necessarily in total input dimension $p$. 
Figure~\ref{fig:deepGP} shows that for a single-layer deep GP (which is essentially a typical GP) with anisotropic lengthscales, its training time increases strongly with $p$ because a different parameter is needed for each input dimension. As discussed in our theoretical results, our approach adapts to anisotropy and hence this figure shows only models that also allow for anisotropy. However, we also found in this same study that for a GP with isotropic lengthscales, its training time is roughly constant with $p$ because the same parameter is used for all input dimensions; hence such a GP could be computationally useful for isotropic regression functions with high-dimensional input variables.

\begin{figure}[h]
\centering
\includegraphics[width=\textwidth]{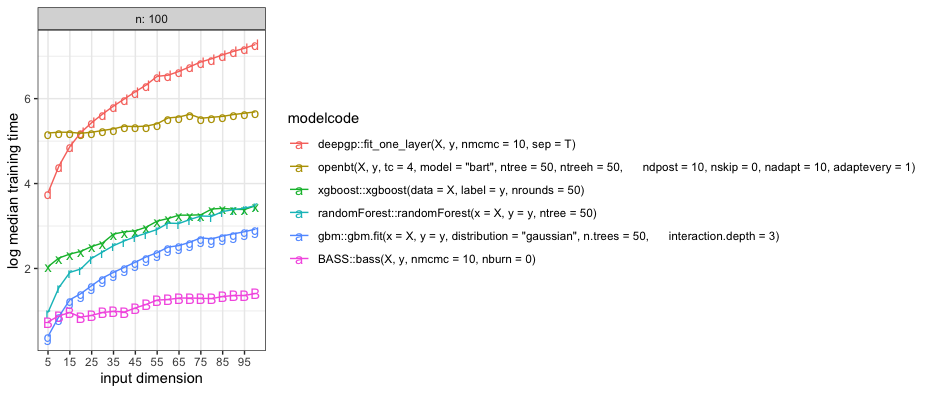}
\caption{Average training time (over 10 evaluations per input dimension) of metamodel implementations on an M1-chip 4-core laptop for $n=100$.}
\label{fig:deepGP}
\end{figure}

Figure~\ref{fig:n3000} shows the training time for 
the \texttt{OpenBT} implementation of BART \cite{openbt}, 
the \texttt{randomForest} package \citep{randomForest2002} implementation of random forests \cite{Breiman01}, 
the \texttt{xgboost} package \citep{xgboost2024} implementation of Extreme Gradient Boosting, 
the \texttt{gbm} package \citep{gbm2024} implementation of Generalized Boosted Regression Models, 
and the \texttt{BASS} package \citep{bass2020} implementation of Bayesian MARS \citep{francom2018sensitivity}.
We see that as $p$ increases, \texttt{randomForest}, \texttt{xgboost}, and \texttt{gbm} seem to increase at roughly the same rate, and that this rate is larger than the rate for either \texttt{openBT} or \texttt{BASS}.
The training time for \texttt{openBT} seems to have a larger start-up time, but based on Figure~\ref{fig:n3000} and Figure~\ref{fig:deepGP}, the training time for \texttt{openBT} seems to grow more slowly in $n$ than do the other ensemble methods.
Interestingly, \texttt{BASS} has a small training time for both sample sizes, and seems to grow slowly in $p$.

\begin{figure}[h]
\centering
\includegraphics[width=\textwidth]{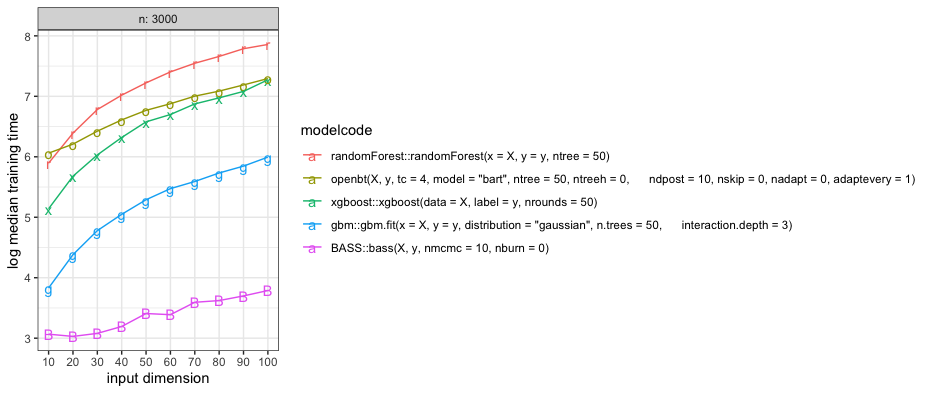}
\caption{Average training time (over 10 evaluations per input dimension) of metamodel implementations on an M1-chip 4-core laptop for $n=3000$.}
\label{fig:n3000}
\end{figure}

\end{document}